\documentclass[11pt,a4paper]{article}
\pdfoutput=1

\usepackage{amssymb}
\usepackage{amsmath}
\usepackage{amsfonts}
\usepackage{bbm}
\usepackage{amsthm}
\usepackage{mathrsfs}
\usepackage{hyperref}
\usepackage{color}
\usepackage[margin=2.41cm]{geometry}
\usepackage[all,cmtip,2cell]{xy}
\UseTwocells

\usepackage[utf8]{inputenc}
\usepackage{graphicx}
\usepackage{varwidth}
\usepackage{comment}
\usepackage{faktor}
\usepackage[shortlabels]{enumitem}
\usepackage{euscript}

\usepackage{mathtools}

\usepackage{upgreek}
\usepackage{rotating}

\usepackage{tikz}
\usetikzlibrary{shapes.geometric}

\usepackage{tikz-cd}
\usetikzlibrary{matrix,arrows,calc,decorations.pathmorphing,fit,shapes.geometric,decorations.pathreplacing,positioning}

\usepackage{amssymb}
\tikzset{curve/.style={settings={#1},to path={(\tikztostart)
    .. controls ($(\tikztostart)!\pv{pos}!(\tikztotarget)!\pv{height}!270:(\tikztotarget)$)
    and ($(\tikztostart)!1-\pv{pos}!(\tikztotarget)!\pv{height}!270:(\tikztotarget)$)
    .. (\tikztotarget)\tikztonodes}},
    settings/.code={\tikzset{quiver/.cd,#1}
        \def\pv##1{\pgfkeysvalueof{/tikz/quiver/##1}}},
    quiver/.cd,pos/.initial=0.35,height/.initial=0}

\definecolor{darkred}{rgb}{0.8,0.1,0.1}

\hypersetup{
	colorlinks=true,         
	urlcolor=darkred, 
	linkcolor=darkred,
	citecolor=blue,
}

\theoremstyle{plain}
\newtheorem{theo}{Theorem}[section]
\newtheorem{lem}[theo]{Lemma}
\newtheorem{propo}[theo]{Proposition}
\newtheorem{cor}[theo]{Corollary}

\theoremstyle{definition}
\newtheorem{defi}[theo]{Definition}

\newtheorem{assu}[theo]{Assumption}

\newtheorem{exx}[theo]{Example}
\newenvironment{ex}
{\pushQED{\qed}\exx}
{\popQED\endexx}

\newtheorem{remm}[theo]{Remark}
\newenvironment{rem}
{\pushQED{\qed}\remm}
{\popQED\endremm}

\newtheorem{constrr}[theo]{Construction}
\newenvironment{constr}
{\pushQED{\qed}\constrr}
{\popQED\endconstrr}

\numberwithin{equation}{section}

\renewcommand{\emptyset}{\varnothing}

\def\nn{\nonumber}

\def\bbK{\mathbb{K}}
\def\bbR{\mathbb{R}}
\def\bbC{\mathbb{C}}

\def\bbZ{\mathbb{Z}}

\def\bbM{\mathbb{M}}
\def\bbE{\mathbb{E}}

\def\id{\mathrm{id}}

\def\dd{\mathrm{d}}

\def\1{I}
\def\oone{\mathbbm{1}}
\def\op{\mathrm{op}}

\def\Set{\mathbf{Set}}
\def\Alg{\mathbf{Alg}}

\def\Ch{\mathbf{Ch}}

\def\CastAlg{C^{\ast}\mathbf{Alg}_\bbC}

\def\CastCat{C^{\ast}\mathbf{Cat}_\bbC}

\def\CC{\mathbf{C}}

\def\TT{\mathbf{T}}
\def\Cat{\mathbf{Cat}}
\def\CAT{\mathbf{CAT}}

\def\Grpd{\mathbf{Grpd}}

\def\Op{\mathbf{Op}}

\def\AQFT{\mathbf{AQFT}}

\def\Fun{\mathbf{Fun}}

\def\AAlg{\EuScript{A}\mathsf{lg}}

\def\Ad{\mathrm{Ad}}

\def\Rep{\mathbf{Rep}}
\def\SSS{\mathbf{SSS}}

\def\AAA{\mathfrak{A}}

\def\FFF{\mathfrak{F}}

\def\P{\mathcal{P}}

\def\V{\EuScript{V}}
\def\CCastCat{C^\ast\EuScript{C}\mathsf{at}_{\bbC}}
\def\CConv{\EuScript{C}\mathsf{onv}}

\def\colim{\mathrm{colim}}
\def\hocolim{\mathrm{hocolim}}
\def\bilim{\mathrm{bilim}}

\def\pt{\mathrm{pt}}

\newcommand\und[1]{\underline{#1}}

\DeclareMathOperator{\cdiamond}{\diamond}

\DeclareMathOperator*{\Motimes}{\text{\raisebox{0.25ex}{\scalebox{0.8}{$\bigotimes$}}}}

\newcommand{\perpline}[3]{%
  \begin{tikzcd}[
    baseline={([yshift=-axis_height]mycd)},
    ampersand replacement=\&,
    every matrix/.append style={
      name=mycd,
      nodes={inner sep=0pt, outer sep=2pt}
      },
    row sep=0pt, column sep=12pt
    ]
    \& #2 \&\\
    #1\ar[to=mycd-1-2.west]\& \scriptstyle{\text{$\perp$}} \& #3 \ar[to=mycd-1-2.east]
  \end{tikzcd}%
  }

\def\sk{\vspace{2mm}}

\makeatletter
\let\@fnsymbol\@alph
\makeatother

\title{%
Prefactorization algebras of superselection sectors 
}

\author{%
Marco Benini$^{1,2,a}$, Victor Carmona$^{3,b}$\ and\ Alexander Schenkel$^{4,c}$\vspace{4mm}\\
{\small ${}^1$ Dipartimento di Matematica, Dipartimento di Eccellenza 2023-27, Universit\`a di Genova,}\\
{\small Via Dodecaneso 35, 16146 Genova, Italy.}\vspace{2mm}\\
{\small ${}^2$ INFN, Sezione di Genova,}\\
{\small Via Dodecaneso 33, 16146 Genova, Italy.}\vspace{2mm}\\
{\small ${}^3$ Max Planck Institut f\"ur Mathematik in den Naturwissenschaften,}\\
{\small Inselstra\ss e 22, 04103 Leipzig, Germany.}\vspace{2mm}\\
{\small ${}^4$ Dipartimento di Matematica, Universit{\`a} di Trento and INFN-TIFPA,}\\
{\small Via Sommarive 14, 38123 Povo (Trento), Italy.}\vspace{4mm}\\
{\small \begin{tabular}{ll}
Email: & ${}^a$~\href{mailto:marco.benini@unige.it}{\texttt{marco.benini@unige.it}}\\
& ${}^b$~\href{mailto:victor.carmona@mis.mpg.de}{\texttt{victor.carmona@mis.mpg.de}}\\
& ${}^c$~\href{mailto:alexander.schenkel@unitn.it}{\texttt{alexander.schenkel@unitn.it}}
\vspace{2mm}
\end{tabular}
}
}

\date{April 2026}

\begin{document}

\maketitle

\begin{abstract}
\noindent This paper revisits the theory of superselection sectors in algebraic quantum field theory from the modern perspective of prefactorization algebras. Under the standard assumptions of Haag duality and a locally faithful vacuum representation, it is shown that every AQFT defined over a filtered orthogonal category of spacetime regions, satisfying some mild additional geometric hypotheses, has an associated locally constant $C^\ast$-categorical prefactorization algebra of superselection sectors over the same orthogonal category. In the case of double cones in the $(n\geq 2)$-dimensional Minkowski spacetime, our approach provides a conceptual explanation for the well-known $\mathbb{E}_n$-monoidal structure on the $C^\ast$-category of superselection sectors as the combination, through Dunn-Lurie additivity $\mathbb{E}_n\simeq \mathbb{E}_1\otimes \mathbb{E}_{n-1}$, of the familiar $\mathbb{E}_1$-monoidal structure from Haag duality and an $\mathbb{E}_{n-1}$-monoidal structure from Lorentzian geometry. A refinement of our results to equivariant contexts under a discrete group $G$ is also provided.
\end{abstract}
\vspace{-1mm}

\vspace{-2mm}

\vspace{-2mm}

\renewcommand{\baselinestretch}{0.85}\normalsize
\tableofcontents
\renewcommand{\baselinestretch}{1.0}\normalsize

\newpage


\section{Introduction and summary}
The theory of superselection sectors is a prominent and time-honored topic 
in algebraic quantum field theory (AQFT) whose principal goal is 
to extract the topological charges associated 
with a (not necessarily topological) quantum field theory. In the original
Doplicher-Haag-Roberts approach \cite{DHR}, such superselection sectors
are described by a specific class of representations of the AQFT $\AAA$
which behave like topological objects as they are required to satisfy
localizability and transportability properties relative to the underlying spacetime
geometry and a fixed vacuum representation $\pi_0$. A key result in this context is that
such representations form a monoidal $C^\ast$-category $\SSS_{(\AAA,\pi_0)}$
which, depending on the dimension of the underlying spacetime,
can be endowed further with the structure of a braided or symmetric monoidal $C^\ast$-category.
See also \cite{Roberts1,Roberts2,Halvorson,SSS} for reviews and 
more modern presentations of the theory of superselection sectors.
\sk

The aim of the present paper is to revisit the theory of superselection
sectors from a modern perspective using contemporary and powerful techniques
from algebraic topology. It is by now well-understood that various braided
or symmetric monoidal categories arising in applications, especially
in quantum field theory, have a direct and interesting topological 
origin rooted in the spaces of mutually disjoint embeddings of families of $n$-dimensional disks into a single disk.
Such structures are known as \textit{little $n$-disk operads}
$\bbE_n$, see e.g.\ \cite{Horel} and \cite{LurieHA}, which, when represented on $1$-categories, 
give rise to monoidal (=$\bbE_1$), braided monoidal (=$\bbE_2$) and symmetric monoidal (=$\bbE_{\geq 3}$) categories.
It is therefore natural to expect that the braided or symmetric monoidal structures on the
$C^\ast$-categories of superselection sectors in AQFT arise through a similar 
topological mechanism. We shall explain in this paper that this is indeed the case.
\sk

To establish a link between the more structured, and therefore also more rigid, Lorentzian geometric regions featuring
in AQFT, such as double cones in the Minkowski spacetime, and the 
topologically behaving little disks of the $\bbE_n$-operad, we use the framework
of \textit{prefactorization algebras} which originated in the works of Costello and Gwilliam
\cite{CG1,CG2}. Prefactorization algebras are a very general class of algebraic objects of geometric origin,
covering a diverse range of contexts, reaching from topological over complex
to Riemannian and Lorentzian geometries. Most importantly for the context of this paper, 
algebras over the little $n$-disk operad $\bbE_n$ can be identified 
with \textit{locally constant} prefactorization algebras on $\bbR^n$, 
see e.g.\ \cite[Theorem 5.4.5.9]{LurieHA}, and AQFTs admit an equivalent description
in terms of prefactorization algebras over their underlying orthogonal category $\CC^\perp$ of spacetime
regions, see e.g.\ \cite[Theorem 2.9]{BPSWcategorified} and also Example \ref{ex:PFA} in the main text.
Therefore, the theory of prefactorization algebras provides us with precisely the right framework and
tools to search for and establish links between AQFT and $\bbE_n$-algebras.
\sk

The main result of the present paper is Theorem \ref{theo:PFAstructure},
which proves that, under the standard assumptions of Haag duality
and a locally faithful vacuum representation $\pi_0$, the $C^\ast$-categories of superselection sectors of
an AQFT $\AAA\in\AQFT(\CC^\perp)$ over a filtered\footnote{Our reason behind this 
rather restrictive assumption of \textit{filtered} orthogonal categories $\CC^\perp$ is 
that it provides us with a direct and simple description of superselection sectors
in terms of localizable and transportable $\ast$-endomorphisms \cite{DHR,Halvorson,SSS}.
It is currently not clear to us if the results of the present paper 
can be generalized to superselection sector theory beyond this filtered context, 
which is usually formulated in terms of Roberts' net cohomology, see e.g.\ \cite{Roberts2,BrunettiRuzzi}.} 
orthogonal category $\CC^\perp$ (satisfying some mild additional geometric hypotheses) 
carry the structure of a locally constant prefactorization algebra 
\begin{flalign}\label{eqn:intro}
\SSS_{(\AAA,\pi_0)}\,\in\,\Alg_{\P_{\CC^\perp}}^{\mathrm{l.c.}}\!\Big(\Alg_{\mathsf{uAs}}\big(\CastCat\big)\Big)
\end{flalign}
over the same orthogonal category $\CC^\perp$. This prefactorization algebra 
assigns to every spacetime region $U\in \CC^\perp$ the (strict) monoidal $C^\ast$-category
$\SSS_{(\AAA,\pi_0)}(U)\in \Alg_{\mathsf{uAs}}\big(\CastCat\big)$ consisting
of all superselection sectors which are also strictly localized in the region $U$.
Local constancy means that the monoidal $\ast$-functor $\SSS_{(\AAA,\pi_0)}(U)\stackrel{\sim}{\longrightarrow} 
\SSS_{(\AAA,\pi_0)}(V)$ associated with any embedding $U\to V$ of spacetime regions in $\CC^\perp$ is an equivalence,
even if $V$ is geometrically larger than $U$, 
witnessing the intrinsic topological behavior of superselection sectors. 
It is worthwhile to highlight that the $C^\ast$-categories of superselection sectors in
\eqref{eqn:intro} carry two distinct but compatible types of algebraic structures, namely a 
locally constant prefactorization algebra structure $\Alg_{\P_{\CC^\perp}}^{\mathrm{l.c.}}$ 
and an object-wise monoidal structure $\Alg_{\mathsf{uAs}}$. The former is of a geometric
origin in the underlying orthogonal category $\CC^\perp$ of spacetime regions, while the
latter is of an analytic origin in Haag duality and its associated standard 
monoidal structure on the category of sectors. Our main result admits a relatively straightforward
generalization and refinement to the case in which the AQFT $\AAA$ is endowed with 
an additional $G$-equivariance structure, for $G$ a discrete group, 
which induces a $G$-equivariance structure
on both the associated locally constant prefactorization algebra 
of all superselection sectors and that of covariant superselection sectors, 
see Theorems \ref{theo:SSSeqvPFA} and \ref{theo:covSSSeqvPFA}.
This provides an alternative perspective on and implementation of 
the recent equivariant superselection sector theory in \cite{CRV}.
\sk

As a concrete application, we will consider in Section \ref{sec:examples}
the typical example of AQFTs defined on double cones in the $(n\geq 2)$-dimensional
Minkowski spacetime $\bbM^n$ and provide an elegant topological argument
explaining why, in this simple geometric context, our locally constant 
prefactorization algebras of superselection sectors 
from Theorem \ref{theo:PFAstructure} admit an equivalent
yet simpler description in terms of $\bbE_n$-monoidal $C^\ast$-categories.
The origin of this equivalence is rooted in the homotopy type of the 
configuration spaces of causally disjoint points in double cones in
the Minkowski spacetime, providing a direct link
between the Lorentzian geometry of Minkowski spacetime and $\bbE_n$-operads.
Such a link was established earlier in \cite{Grady} by means of a construction 
that is, however, not intrinsically Lorentzian in nature.
\sk

The paper is organized as follows: In Section \ref{sec:prelim} we briefly recall
and summarize some basic aspects of the operadic description of
prefactorization algebras and AQFTs, including their $G$-equivariant generalizations.
In Section \ref{sec:SSS} we show that the standard definitions and 
constructions used in the theory of superselection sectors
assemble very naturally into a locally constant prefactorization algebra 
structure on the monoidal $C^\ast$-categories of strictly localized superselection sectors, 
culminating in our main Theorem \ref{theo:PFAstructure}.
In Section \ref{sec:dEqvSSS} we provide a $G$-equivariant generalization
of these constructions and results, leading to Theorems \ref{theo:SSSeqvPFA} and \ref{theo:covSSSeqvPFA}.
The purpose of Section \ref{sec:examples} is to apply our general framework 
to the typical example of AQFTs defined on double cones in the $(n\geq 2)$-dimensional
Minkowski spacetime $\bbM^n$, which leads to a conceptual and topological explanation
for the braided monoidal (for $n=2$)
or symmetric monoidal (for $n\geq 3$) structure on the $C^\ast$-category
of superselection sectors on double cones in the $n$-dimensional
Minkowski spacetime $\bbM^n$. Appendix \ref{app:technical}
provides the necessary technical arguments about $\infty$-categories
and $\infty$-operads which are required to prove the results of Section \ref{sec:examples}.


\section{\label{sec:prelim}Preliminaries}
In this section we provide a concise recapitulation
of the operadic description of prefactorization algebras 
and algebraic quantum field theories (AQFTs). We refer
the reader to \cite{CG1,CG2,BSWoperad,BPSWcategorified,Carmona,Yau}
for more detailed expositions.
\sk

Recall that an \textit{orthogonal category} $\CC^\perp := (\CC,\perp)$
is a pair consisting of a small category $\CC$ 
and an orthogonality relation $\perp$, which is
a chosen subset of the set of cospans $\{f_1 :  U_1\to V \leftarrow U_2 : f_2\}$
in $\CC$ that is closed under transposition and under
pre- and post-compositions with $\CC$-morphisms. 
We write $f_1\perp f_2$ whenever a cospan belongs to $\perp$. Orthogonal categories are an abstraction 
of categories of spacetimes which are endowed with a
notion of independent pairs of subspacetimes.
\begin{defi}\label{def:Poperad}
The \textit{prefactorization operad} $\P_{\CC^\perp}$ associated to an
orthogonal category $\CC^\perp$ is the $\Set$-valued colored symmetric operad defined
by the following data:
\begin{enumerate}[(1)]
\item The objects of $\P_{\CC^\perp}$ are the objects of the category $\CC$.

\item The set of operations from a tuple of objects 
$\und{U} :=(U_1,\dots,U_n)\in \CC^n$ to an object $V\in\CC$ is given by
\begin{flalign}
\P_{\CC^\perp}\big(\substack{V \\ \und{U}}\big)\,:=\, \bigg\{ \und{f} := (f_1,\dots,f_n)\in \prod_{i=1}^n \CC(U_i,V)\,:\, f_i\perp f_j~ \forall i\neq j\,\bigg\}\quad.
\end{flalign}
For the empty tuple $\und{U}=()$,
we set $\P_{\CC^\perp}\big(\substack{V \\ ()}\big):=\{\pt_V^{}\}$ to be a singleton.

\item The operadic composition maps 
\begin{subequations}
\begin{flalign}
\gamma \,:\, \P_{\CC^\perp}\big(\substack{V \\ \und{U}}\big)\times\prod_{i=1}^n 
\P_{\CC^\perp}\big(\substack{U_i \\ \und{W_i}}\big)~\xrightarrow{\;\quad\;}~ \P_{\CC^\perp}\big(\substack{V \\ \und{\und{W}}}\big)\quad,
\end{flalign}
where $\und{\und{W}} := (\und{W_1},\dots,\und{W_n})$ denotes the concatenation of tuples, 
are given by the following compositions in the category $\CC$
\begin{flalign}
\gamma \big(\und{f},(\und{g_1},\dots,\und{g_n}) \big)\,:=\, \und{f}\,\und{\und{g}} \,:=\, \big(f_1\,g_{11},\dots, f_1\, g_{1 k_1},\dots,f_n\, g_{n1},\dots,f_n \, g_{n k_n}\big) \quad.
\end{flalign}
\end{subequations}

\item The identity operations are $\id_V \in \P_{\CC^\perp}\big(\substack{V \\ V}\big)$.

\item The permutation actions $\P_{\CC^\perp}(\sigma) : \P_{\CC^\perp}\big(\substack{V \\ \und{U}}\big)
\to \P_{\CC^\perp}\big(\substack{V \\ \und{U}\sigma}\big)$, for $\sigma\in\Sigma_n$, are given by
\begin{flalign}
\P_{\CC^\perp}(\sigma)(\und{f}) \,:=\, \und{f}\sigma \,:=\, (f_{\sigma(1)},\dots, f_{\sigma(n)})\quad.
\end{flalign}
\end{enumerate}
\end{defi}
\begin{defi}\label{def:PFA}
For $\CC^\perp$ an orthogonal category and $\TT$ a symmetric monoidal category, 
the category of \textit{$\TT$-valued prefactorization algebras over $\CC^\perp$}
is defined as the category $\Alg_{\P_{\CC^\perp}}^{}\!\big(\TT\big)$
of algebras over the operad $\P_{\CC^\perp}$ from Definition \ref{def:Poperad}
with values in $\TT$.
\end{defi}

\begin{rem}\label{rem:PFA}
Unpacking this definition, one finds that 
a prefactorization algebra $\FFF\in \Alg_{\P_{\CC^\perp}}^{}\!\big(\TT\big)$
consists of the following data:
\begin{itemize}
\item[(1)] For every object $U\in\CC$, an object $\FFF(U)\in\TT$.

\item[(2)] For every operation $\und{f} = (f_1,\dots,f_n) : \und{U}\to V$ in $\P_{\CC^{\perp}}$, a $\TT$-morphism
\begin{flalign}
\xymatrix{
\FFF(\und{f}) \,:\, \FFF(\und{U})\,:=\,\bigotimes\limits_{i=1}^n \FFF(U_i)~\ar[r]
&~\FFF(V)
}\quad.
\end{flalign}
\end{itemize}
These data have to satisfy $\FFF(\id_V) = \id_{\FFF(V)}$, for all $V\in \CC$,
\begin{equation}
\begin{tikzcd}[column sep=large]
\ar[rd,"\FFF(\und{f}\,\und{\und{g}})"'] \FFF(\und{\und{W}}) \ar[r,"\FFF(\und{\und{g}})"] & 
\FFF(\und{U}) \ar[d, "\FFF(\und{f})"]\\
 &  \FFF(V)
\end{tikzcd}
\quad,
\end{equation}
for all composable operations $\und{\und{g}} : \und{\und{W}}\to\und{U}$ and $\und{f} : \und{U}\to V$ in
$\P_{\CC^{\perp}}$, and
\begin{equation}
\begin{tikzcd}[column sep=small]
\ar[rd,"\FFF(\und{f})"'] \FFF(\und{U}) \ar[rr,"\text{permute}","\cong"'] && \FFF(\und{U}\sigma) \ar[dl, "\FFF(\und{f}\sigma)"]\\
 &  \FFF(V)&
\end{tikzcd}
\quad,
\end{equation}
for all operations $\und{f} : \und{U}\to V$ in $\P_{\CC^{\perp}}$ and all permutations $\sigma\in\Sigma_n$.
\sk

A morphism $\zeta : \FFF\to \FFF^\prime$ in $\Alg_{\P_{\CC^\perp}}^{}\!\big(\TT\big)$ is a family
of $\TT$-morphisms $\zeta_U : \FFF(U)\to\FFF^\prime(U)$, for all $U\in \CC$, which is compatible
with the structure maps, i.e.\
\begin{equation}
\begin{tikzcd}[column sep=large]
\ar[d,"\FFF(\und{f})"'] \FFF(\und{U}) \ar[r,"\zeta_{\und{U}}"] & \FFF^\prime(\und{U}) \ar[d, "\FFF^\prime(\und{f})"]\\
\FFF(V) \ar[r,"\zeta_V"']&  \FFF^\prime(V)
\end{tikzcd}
\quad,
\end{equation}
for all operations $\und{f} : \und{U}\to V$ in $\P_{\CC^{\perp}}$.
\end{rem}

\begin{ex}\label{ex:PFA}
We would like to highlight that Definition \ref{def:PFA} is very 
general as it captures a variety of concepts appearing throughout the literature.
For instance:
\begin{itemize}
\item[(1)] Choosing for $\TT=\Ch_\bbC$ the symmetric monoidal category of cochain complexes of vector spaces
and for $\CC^\perp=\mathbf{Open}(M)^\perp$ the poset of open subsets of a manifold $M$ with orthogonality
relation determined by disjointness of subsets, i.e.\ 
$(U_1\subseteq V)\perp(U_2\subseteq V)~\Longleftrightarrow~U_1\cap U_2 =\varnothing$, 
one recovers the prefactorization algebras of Costello and Gwilliam \cite{CG1,CG2}.

\item[(2)] Choosing for $\TT= \CastAlg$ the symmetric monoidal category of $C^\ast$-algebras
(with the maximal tensor product, see e.g.\ \cite[Chapter 6.3]{Murphy}), one obtains a category
\begin{flalign}
\Alg_{\P_{\CC^\perp}}^{}\!\big(\CastAlg\big)\,\simeq\, \AQFT(\CC^\perp)
\end{flalign}
which is equivalent to the category of $C^\ast$-algebraic AQFTs over $\CC^\perp$.
The latter is the full subcategory of the functor category 
$\Fun(\CC,\CastAlg)$ on those functors $\AAA: \CC \to \CastAlg$ which are 
$\perp$-commutative, i.e.\ such that 
$\mu \circ \big(\AAA(f_1) \otimes \AAA(f_2)\big) = \mu^\op \circ \big(\AAA(f_1) \otimes \AAA(f_2)\big)$, 
for all $f_1 \perp f_2$ in $\CC^\perp$, where $\mu^{(\op)}$ 
denotes the (opposite) $C^\ast$-algebra multiplication.
This was shown in \cite[Theorem 2.9]{BPSWcategorified} for the case of plain algebras
over a field $\bbK$, but the same argument applies in the context of $C^\ast$-algebras as well.
The key point is that the compatibility between the $\P_{\CC^\perp}^{}$-algebra
structure and the object-wise $C^\ast$-algebra structures implies the $\perp$-commutativity
axiom of AQFT via an Eckmann-Hilton argument.

\item[(3)] Choosing for $\TT=\CastCat$ the symmetric monoidal category of $C^\ast$-categories
(with the maximal tensor product, see e.g.\ \cite{CastCat}), one recovers the concept
of $C^\ast$-categorical prefactorization algebras that appeared recently
in the context of operator algebraic factorization homology \cite{Hataishi}
and topological order in lattice quantum systems \cite{BCNSsectors}. \qedhere
\end{itemize}
\end{ex}

As justified by Example \ref{ex:PFA} (2), we shall work in our present 
paper with the following equivalent definition of AQFTs which 
is expressed in terms of prefactorization algebras.
\begin{defi}\label{def:AQFT}
For $\CC^\perp$ an orthogonal category, the category of
($C^\ast$-algebraic) \textit{AQFTs over $\CC^\perp$} is defined by
\begin{flalign}
\AQFT\big(\CC^\perp\big)\,:=\, \Alg_{\P_{\CC^\perp}}^{}\!\big(\CastAlg\big)\quad.
\end{flalign}
\end{defi}

The above definitions can be extended easily to an equivariant setting.
Let $G$ be a discrete group and denote by $\mathsf{B}G \in\Grpd$
its delooping, i.e.\ the groupoid consisting of a single object 
and morphisms $G$, with identity and composition given by the group structure on $G$.
A \textit{$G$-action on an orthogonal category} is a $2$-functor
$\CC^\perp : \mathsf{B}G\to \Cat^\perp$ to the $2$-category of orthogonal categories,
orthogonal functors and natural transformations. Explicitly, that is the datum
of an orthogonal category $\CC^\perp$ and a family of 
orthogonal functors $\{\alpha_g : \CC^\perp\to \CC^\perp\}_{g\in G}$
which satisfy $\alpha_e = \id_{\CC^\perp}$, for the unit element $e\in G$, and
$\alpha_{g^\prime}\,\alpha_g = \alpha_{g^\prime g}$, for all $g,g^\prime\in G$.
Using the elementary fact that the assignment of the prefactorization operads 
from Definition \ref{def:Poperad} is canonically $2$-functorial $\P_{(-)}: \Cat^\perp\to \Op$,
and so is the assignment $\Alg_{(-)}\!\big(\TT\big) : \Op^\op\to \CAT$ of the categories of operad algebras,
we can state the following definition.
\begin{defi}\label{def:eqvPFA}
For $\CC^\perp : \mathsf{B}G\to \Cat^\perp$ an orthogonal category with $G$-action
and $\TT$ a symmetric monoidal category, 
the category of \textit{$\TT$-valued $G$-equivariant prefactorization algebras
over $\CC^\perp$} is defined as the bicategorical limit
\begin{flalign}
\Alg_{\P_{\CC^\perp}}^G\!\big(\TT\big) \,:=\,
\bilim\bigg(
\xymatrix{
\big(\mathsf{B}G\big)^\op \ar[r]^-{\CC^\perp}~&~\big(\Cat^\perp\big)^\op \ar[rr]^-{\Alg_{\P_{(-)}}\!(\TT)}~&&~\CAT
}\bigg)
\end{flalign}
in the $2$-category $\CAT$ of categories, functors and natural transformations.
\end{defi}

\begin{rem}\label{rem:eqvPFA}
Unpacking this definition, one finds that 
a $G$-equivariant prefactorization algebra $\FFF\in \Alg_{\P_{\CC^\perp}}^G\!\big(\TT\big)$
consists of the following data:
\begin{itemize}
\item[(1)] A prefactorization algebra $\FFF\in \Alg_{\P_{\CC^\perp}}^{}\!\big(\TT\big)$ over the underlying
orthogonal category $\CC^\perp$, see also Remark \ref{rem:PFA} for an explicit description.

\item[(2)] For every $g\in G$, an $\Alg_{\P_{\CC^\perp}}^{}\!\big(\TT\big)$-isomorphism
\begin{subequations}
\begin{flalign}
\xymatrix{
\Psi_g\,:\, \FFF \ar[r]^-{\cong}~&~\alpha_g^\ast(\FFF)
}
\end{flalign}
to the prefactorization algebra defined by 
\begin{flalign}
\xymatrix{
\alpha_g^\ast(\FFF) \,:\,\P_{\CC^\perp} \ar[r]^-{\P_{\alpha_g}}~&~\P_{\CC^\perp}\ar[r]^-{\FFF}~&~\TT
}\quad.
\end{flalign}
\end{subequations}
\end{itemize}
These data have to satisfy $\Psi_e = \id$, for the unit element $e\in G$, and
\begin{equation}
\begin{tikzcd}
\FFF \ar[r,"\Psi_g"] \ar[d,"\Psi_{g^\prime g}"'] & \alpha_g^\ast(\FFF)\ar[d, "\alpha_{g}^\ast(\Psi_{g^\prime})"]\\
\alpha_{g^\prime g}^\ast(\FFF)\ar[r,equal] &  \alpha_g^\ast \alpha_{g^\prime}^\ast(\FFF)
\end{tikzcd}
\quad,
\end{equation}
for all $g,g^\prime\in G$.
\sk

A morphism $\zeta : (\FFF,\Psi)\to (\FFF^\prime,\Psi^\prime)$ in $\Alg_{\P_{\CC^\perp}}^G\!\big(\TT\big)$ 
is an $\Alg_{\P_{\CC^\perp}}^{}\!\big(\TT\big)$-morphism $\zeta : \FFF\to\FFF^\prime$, see also Remark 
\ref{rem:PFA} for an explicit description, which satisfies
\begin{equation}
\begin{tikzcd}[column sep=large]
\ar[d,"\Psi_g"'] \FFF \ar[r,"\zeta"] & \FFF^\prime \ar[d, "\Psi^\prime_g"]\\
\alpha_g^\ast(\FFF) \ar[r,"\alpha_g^\ast(\zeta)"']&  \alpha_g^\ast(\FFF^\prime)
\end{tikzcd}
\quad,
\end{equation}
for all $g\in G$.
\end{rem}

Analogously to Definition \ref{def:AQFT}, this definition can be specialized to the context of AQFT.
\begin{defi}\label{def:eqvAQFT}
For $\CC^\perp : \mathsf{B}G\to \Cat^\perp$ an orthogonal category with $G$-action, 
the category of ($C^\ast$-algebraic) \textit{$G$-equivariant AQFTs over $\CC^\perp$} is defined by
\begin{flalign}
\AQFT^G\big(\CC^\perp\big)\,:=\, \Alg_{\P_{\CC^\perp}}^{G}\!\big(\CastAlg\big)\quad.
\end{flalign}
\end{defi}


\section{\label{sec:SSS}Prefactorization algebras of superselection sectors}
The aim of this section is to prove that, under suitable conditions to be specified
below, the $C^\ast$-categories of localized superselection sectors of an AQFT, see e.g.\ \cite{DHR,Roberts2,SSS},
carry the structure of a locally constant prefactorization algebra taking values in strict monoidal $C^\ast$-categories. 
Throughout this section we fix an orthogonal category $\CC^\perp$ 
and assume that its underlying category $\CC$ is filtered.

\paragraph{Superselection sectors:} 
Given any AQFT over $\CC^\perp$ as in Definition \ref{def:AQFT}, i.e.\
\begin{flalign}
\AAA\,\in\,\AQFT\big(\CC^\perp\big)\quad,
\end{flalign}
our filteredness assumption on $\CC^\perp$ implies that its representation category
(in the sense of \cite{BrunettiRuzzi}) is equivalent to the category of
representations of the associated universal $C^\ast$-algebra
\begin{flalign}\label{eqn:universal}
\AAA_{\star}\,:=\,\colim\Big(
\AAA\big\vert \,:\,\CC\,\longrightarrow\,\CastAlg\Big)\,\in\,\CastAlg\quad,
\end{flalign}
which is defined as the colimit of the restriction of $\AAA$ 
to the category $\CC\subseteq \P_{\CC^\perp}$ of $1$-ary operations.
Concretely, the representation category of $\AAA$ is given by
\begin{flalign}
\Rep_{\AAA}\,:=\,\begin{cases}
\mathsf{Obj}:\, &\, \text{$\ast$-homomorphisms
$\pi:\AAA_{\star}\to \mathrm{B}(H)$ to the $C^\ast$-algebra}\\
\,&\,\text{of bounded operators on a complex Hilbert space $H$}\sk\\
\mathsf{Mor}:\,&\,\text{$T : \pi\to \pi^\prime$ are bounded operators $T : H\to H^\prime$}\\ 
\,&\, \text{such that $T\circ \pi(-) = \pi^\prime(-)\circ T$}
\end{cases}\qquad.
\end{flalign}
Note that $\Rep_\AAA$ is naturally a $C^\ast$-category with involution 
$(T:H\to H^\prime) \mapsto (T^\ast: H^\prime\to H)$ 
given by taking adjoints of bounded operators
and norm given by the operator norm $\vert\!\vert T\vert\!\vert$.
Furthermore, let us observe that any object $\pi \in \Rep_\AAA$
defines, via restriction along the canonical morphisms
$\iota_U : \AAA(U)\to \AAA_\star$ to the colimit \eqref{eqn:universal}, 
representations of the local $C^\ast$-algebras $\pi_U:= \pi\,\iota_U : \AAA(U)\to \mathrm{B}(H)$,
for all $U\in \CC$. These local representations are compatible with each other
in the sense that, for all morphisms $f: U\to \widetilde{U}$ in $\CC$, the diagrams
\begin{equation}
\begin{tikzcd}
&\mathrm{B}(H)& \\
\ar[ru, "\pi_U"] \AAA(U)\ar[rr,"\AAA(f)"'] && \AAA(\widetilde{U})\ar[lu,"\pi_{\widetilde{U}}"']
\end{tikzcd}
\end{equation}
in $\CastAlg$ commute.
\sk

Superselection sectors \cite{DHR,Roberts2,SSS} are a special class
of representations which are defined relative to a given reference
representation $\pi_0\in \Rep_\AAA$, often called the `vacuum representation'.
\begin{defi}\label{def:SSSglobal}
Let $\AAA\in \AQFT\big(\CC^\perp\big)$ be an AQFT over a 
filtered orthogonal category $\CC^\perp$ and $\big(\pi_0:\AAA_\star\to \mathrm{B}(H_0)\big)\in \Rep_\AAA$
a choice of representation. The $C^\ast$-category
of \textit{superselection sectors of $\AAA$ relative to $\pi_0$}
is defined as the full $C^\ast$-subcategory
\begin{flalign}
\SSS_{(\AAA,\pi_0)}\,\subseteq\, \Rep_\AAA
\end{flalign}
consisting of all objects $\big(\pi:\AAA_\star\to \mathrm{B}(H)\big)\in \Rep_\AAA$ which satisfy
the following localizability conditions: For each $U\in \CC$,
there exists a unitary isomorphism $v:H\to H_0$ such that
\begin{flalign}
v\circ \pi_{U^\prime}(-)\,=\, \pi_{0 U^\prime}(-)\circ v\quad,
\end{flalign}
for all orthogonal cospans $(U^\prime \to V)\perp (U\to V)$ in $\CC^\perp$.
\end{defi}

For our purposes, it will be convenient to consider superselection sectors
that are not only localizable in the sense of Definition \ref{def:SSSglobal},
but furthermore are \textit{strictly localized} for some fixed choice of 
object $U\in\CC$.
\begin{defi}\label{def:SSSlocal}
For each $U\in\CC$, we define the $C^\ast$-category
of \textit{strictly $U$-localized superselection sectors of $\AAA$ relative to $\pi_0$}
as the full $C^\ast$-subcategory
\begin{subequations}\label{eqn:SSSlocalintoglobal}
\begin{flalign}
\SSS_{(\AAA,\pi_0)}(U)\,\subseteq\,\SSS_{(\AAA,\pi_0)}
\end{flalign}
consisting of all objects $\big(\pi:\AAA_\star\to \mathrm{B}(H_0)\big)
\in \SSS_{(\AAA,\pi_0)}$ which share the same Hilbert
space with $\pi_0: \AAA_\star\to \mathrm{B}(H_0)$ and satisfy the strict localization condition
\begin{flalign}
\pi_{U^\prime}\,=\, \pi_{0 U^\prime}\quad,
\end{flalign}
\end{subequations}
for all orthogonal cospans $(U^\prime \to V)\perp (U\to V)$ relative
to the fixed object $U\in\CC$.
\end{defi}

\begin{lem}\label{lem:SSSproperties}
\begin{itemize}
\item[(i)] For every morphism $f:U\to \widetilde{U}$ in $\CC$, there is
a sequence of full $C^\ast$-subcategory inclusions
\begin{flalign}
\SSS_{(\AAA,\pi_0)}(U)\,\subseteq\,\SSS_{(\AAA,\pi_0)}(\widetilde{U})\,\subseteq\, \SSS_{(\AAA,\pi_0)}\quad.
\end{flalign}

\item[(ii)] For every $U\in\CC$, the inclusion 
$\SSS_{(\AAA,\pi_0)}(U)\subseteq\SSS_{(\AAA,\pi_0)}$
from Definition \ref{def:SSSlocal}
is a unitary equivalence of $C^\ast$-categories.
\end{itemize}
\end{lem}
\begin{proof}
Item (i): Given any $\pi\in \SSS_{(\AAA,\pi_0)}(U)$, we have to show
that $\pi_{U^\prime}= \pi_{0U^\prime}$, for all
$(U^\prime \to V)\perp (\widetilde{U}\to V)$ in $\CC^\perp$. Since
orthogonality relations are by definition closed under pre-compositions,
we have that $(U^\prime \to V)\perp (U\stackrel{f}{\to}\widetilde{U}\to V)$,
hence the claim follows from the hypothesis that $\pi\in \SSS_{(\AAA,\pi_0)}(U)$.
\sk

Item (ii): Given any $\pi\in \SSS_{(\AAA,\pi_0)}$, there exists 
by Definition \ref{def:SSSglobal} a unitary isomorphism $v : H\to H_0$
such that $v\circ \pi_{U^\prime}(-)=\pi_{0 U^\prime}(-)\circ v$,
for all $(U^\prime \to V)\perp (U\to V)$ in $\CC^\perp$. Then
$\pi^v:= v\circ \pi(-)\circ v^\ast : \AAA_\star\to \mathrm{B}(H_0)$ satisfies
by construction the strict $U$-localization condition
$\pi^v_{U^\prime} = \pi_{0U^\prime}$, for all $(U^\prime \to V)\perp (U\to V)$ in $\CC^\perp$,
hence $\pi^v\in \SSS_{(\AAA,\pi_0)}(U)$. Our claim then follows
from the fact that $v : \pi\to \pi^v$ defines a unitary isomorphism in $\SSS_{(\AAA,\pi_0)}$.
\end{proof} 

Using the $2$-out-of-$3$ property for weak equivalences, it follows from Lemma \ref{lem:SSSproperties}
that the inclusion $\SSS_{(\AAA,\pi_0)}(U)\subseteq\SSS_{(\AAA,\pi_0)}(\widetilde{U})$
is a unitary equivalence too, for all morphisms $f:U\to \widetilde{U}$ in $\CC$.
This yields the following intermediate result.
\begin{propo}\label{prop:SSSfunctor}
The $C^\ast$-categories of strictly localized superselection sectors
from Definition \ref{def:SSSlocal} assemble into a functor
\begin{flalign}
\SSS_{(\AAA,\pi_0)}\,:\, \CC~\longrightarrow~\CastCat
\end{flalign}
which is locally constant in the sense that it assigns
to every morphism $f:U\to \widetilde{U}$ in $\CC$
a unitary equivalence of $C^\ast$-categories.
\end{propo}

\paragraph{Monoidal structures:} To endow the functor
$\SSS_{(\AAA,\pi_0)} : \CC\to\CastCat$ from Proposition
\ref{prop:SSSfunctor} with monoidal 
structures, we have to demand additional assumptions
on the orthogonal category $\CC^\perp$,
the AQFT $\AAA\in \AQFT\big(\CC^\perp\big)$ and 
the reference representation $\pi_0\in \Rep_\AAA$. 
\begin{assu}\label{assu:monoidal}
\begin{itemize}
\item[(1)] There exists an orthogonal cospan
$(U^\prime \to V)\perp (U\to V)$ in $\CC^\perp$,
for every $U\in\CC$.

\item[(2)] The local representations $\pi_{0U} : \AAA(U)\to \mathrm{B}(H_0)$
are faithful, for all $U\in\CC$.

\item[(3)] The AQFT $\AAA\in \AQFT\big(\CC^\perp\big)$ satisfies Haag duality
in the reference representation $\pi_0\in\Rep_\AAA$. Explicitly, for every $U\in\CC$,
we have an equality
\begin{flalign}
\pi_{0U}\big(\AAA(U)\big)^{\prime\prime}\,=\, \bigcap_{\scalebox{0.75}{\text{$\perpline{U^\prime}{V}{U}$}}} 
\pi_{0U^\prime}\big(\AAA(U^\prime)\big)^{\prime}
\end{flalign}
of subsets of $\mathrm{B}(H_0)$, where $S^\prime\subseteq \mathrm{B}(H_0)$ denotes the commutant 
of a subset $S\subseteq \mathrm{B}(H_0)$.
\end{itemize}
\end{assu}

In what follows we denote by
\begin{flalign}
\pi_0(\AAA)^{\prime\prime}\,\in\,\AQFT\big(\CC^\perp\big)
\end{flalign}
the AQFT which assigns to $U\in\CC$ the bicommutant 
$\pi_0(\AAA)^{\prime\prime}(U):= \pi_{0U}\big(\AAA(U)\big)^{\prime\prime}\subseteq \mathrm{B}(H_0)$
and to a morphism $f:U\to \widetilde{U}$ in $\CC$ the inclusion $\pi_{0U}\big(\AAA(U)\big)^{\prime\prime}
\subseteq \pi_{0\widetilde{U}}\big(\AAA(\widetilde{U})\big)^{\prime\prime}$.\footnote{Note that
this defines a $\perp$-commutative functor $\pi_0(\AAA)^{\prime\prime} : \CC\to \CastAlg$,
which is equivalent to an AQFT in the sense of Definition \ref{def:AQFT}, see
Example \ref{ex:PFA} (2).}
The representation $\pi_0$ defines an $\AQFT\big(\CC^\perp\big)$-morphism 
\begin{subequations}
\begin{flalign}
\pi_0 \,:\, \AAA~\longrightarrow~ \pi_0(\AAA)^{\prime\prime}
\end{flalign} 
and its associated $\CastAlg$-morphism
\begin{flalign}
\pi_{0} \,:\, \AAA_\star~\longrightarrow~ \pi_0(\AAA)^{\prime\prime}_\star
\end{flalign}
\end{subequations}
between the universal $C^\ast$-algebras, all denoted
with a slight abuse of notation by the same symbol.

\begin{lem}\label{lem:toolsmonoidal}
Under Assumption \ref{assu:monoidal}, the following statements hold true:
\begin{itemize}
\item[(i)] Let $T : \pi\to \widetilde{\pi}$ be a morphism
in the $C^\ast$-category $\SSS_{(\AAA,\pi_0)}$ from Definition \ref{def:SSSglobal}
with $\pi \in \SSS_{(\AAA,\pi_0)}(U)$ and $\widetilde{\pi} \in \SSS_{(\AAA,\pi_0)}(\widetilde{U})$,
for some $U,\widetilde{U}\in\CC$. For each
(not necessarily orthogonal) cospan $U\to V\leftarrow \widetilde{U}$
in $\CC$, one has that
\begin{flalign}
T\,\in \,\pi_{0V}\big(\AAA(V)\big)^{\prime\prime}\,\subseteq \, \mathrm{B}(H_0)\quad.
\end{flalign}
In particular, since $\CC$ is assumed to be filtered, one always has that
$T\in \pi_0(\AAA)_\star^{\prime\prime}\subseteq \mathrm{B}(H_0)$.

\item[(ii)] For every $U\in \CC$ and $\pi\in\SSS_{(\AAA,\pi_0)}(U)$, there exists
a factorization
\begin{equation}
\begin{tikzcd}
\AAA_\star \ar[rr, "\pi"] \ar[dr, "\pi"'] && \mathrm{B}(H_0)\\
&\pi_0(\AAA)^{\prime\prime}_\star \ar[ru, hookrightarrow] &
\end{tikzcd}
\quad.
\end{equation}

\item[(iii)] For every $U\in \CC$ and $\pi\in\SSS_{(\AAA,\pi_0)}(U)$, there 
exists a unique extension
\begin{equation}
\begin{tikzcd}
\pi_0(\AAA)^{\prime\prime}_\star \ar[rr,dashed,"\rho"] && \pi_0(\AAA)^{\prime\prime}_\star\\
& \ar[ul,"\pi_0"] \AAA_\star \ar[ur, "\pi"'] &
\end{tikzcd}
\end{equation}
which is locally weakly continuous, i.e.\ $\rho_{\widetilde{U}} := \rho \,\iota_{\widetilde{U}} : 
\pi_{0\widetilde{U}}\big(\AAA(\widetilde{U})\big)^{\prime\prime} \to \pi_0(\AAA)^{\prime\prime}_\star$
is weakly continuous, for all $\widetilde{U}\in\CC$. 
\end{itemize}
\end{lem}
\begin{proof}
Item (i): Since $T : \pi\to \widetilde{\pi}$ is a $\SSS_{(\AAA,\pi_0)}$-morphism,
it satisfies by definition $T\circ \pi_{W}(-) = \widetilde{\pi}_{W}(-)\circ T$,
for all $W\in \CC$. Given any $(V^\prime \to W)\perp(V\to W)$ in $\CC^\perp$,
it follows that $(V^\prime\to W)\perp(U\to V\to W)$
and $(V^\prime\to W)\perp (\widetilde{U}\to V\to W)$
since orthogonality relations are closed under pre-compositions. 
Using the strict localization properties of $\pi \in \SSS_{(\AAA,\pi_0)}(U)$ 
and $\widetilde{\pi} \in \SSS_{(\AAA,\pi_0)}(\widetilde{U})$, we compute
\begin{flalign}
T\circ \pi_{0V^\prime}(-)\,=\,T\circ \pi_{V^\prime}(-)\,=\, \widetilde{\pi}_{V^\prime}(-)\circ T
\,=\, \pi_{0V^\prime}(-)\circ T\quad,
\end{flalign}
hence
\begin{flalign}
T\,\in \,\bigcap_{\scalebox{0.75}{\text{$\perpline{V^\prime}{W}{V}$}}} 
\pi_{0V^\prime}\big(\AAA(V^\prime)\big)^{\prime}\,=\, \pi_{0V}\big(\AAA(V)\big)^{\prime\prime}\quad,
\end{flalign}
where the last step uses Haag duality from Assumption \ref{assu:monoidal} (3).
\sk

Item (ii): For each $\widetilde{U}\in \CC$, 
there exists by Assumption \ref{assu:monoidal} (1) an orthogonal cospan 
$(\widetilde{U}^\prime\to V)\perp (\widetilde{U}\to V)$ in $\CC^\perp$.
By Definition \ref{def:SSSglobal}, there exists a unitary isomorphism $v_{\widetilde{U}} : H_0\to H_0$
such that $v_{\widetilde{U}} : \pi\to \pi^{v_{\widetilde{U}}} := v_{\widetilde{U}}\circ \pi(-)\circ 
v_{\widetilde{U}}^\ast$ with $\pi^{v_{\widetilde{U}}}\in \SSS_{(\AAA,\pi_0)}(\widetilde{U}^\prime)$.
Since $\CC$ is assumed to be filtered, there exists a (not necessarily orthogonal) cospan
$U\to W\leftarrow V$ in $\CC$, hence 
$v_{\widetilde{U}} \in \pi_{0W}\big(\AAA(W)\big)^{\prime\prime}
\subseteq \pi_0(\AAA)^{\prime\prime}_\star\subseteq \mathrm{B}(H_0)$ by item (i) of this lemma.
Using also the strict $\widetilde{U}^\prime$-localization property 
$\pi^{v_{\widetilde{U}}}_{\widetilde{U}}= \pi_{0\widetilde{U}}$, it follows that 
\begin{flalign}\label{eqn:tmpfactor}
\pi_{\widetilde{U}}\,=\, v^\ast_{\widetilde{U}}\circ \pi_{0\widetilde{U}}(-)\circ v_{\widetilde{U}}\,:\, 
\AAA(\widetilde{U})~\longrightarrow~\pi_{0W}\big(\AAA(W)\big)^{\prime\prime}\, 
\subseteq\, \pi_0(\AAA)^{\prime\prime}_\star \,\subseteq \, \mathrm{B}(H_0)\quad.
\end{flalign}
Our claim then follows by passing to the colimit $\AAA_\star$.
\sk

Item (iii): We define, for each $\widetilde{U}\in \CC$, the $\ast$-homomorphism
\begin{equation}
\begin{tikzcd}
\pi_{0\widetilde{U}}\big(\AAA(\widetilde{U})\big) \ar[rr,dashed, "\rho_{\widetilde{U}}"] &&
\pi_0(\AAA)_\star^{\prime\prime}\\
&\ar[ul,"\pi_{0\widetilde{U}}","\simeq" sloped ] \AAA(\widetilde{U})\ar[ru,"\pi_{\widetilde{U}}"'] &
\end{tikzcd}
\quad,
\end{equation}
where the upward-left pointing arrow is an isomorphism as a consequence of our faithfulness
Assumption \ref{assu:monoidal} (2). Using \eqref{eqn:tmpfactor}, one finds
that $\rho_{\widetilde{U}} = v^\ast_{\widetilde{U}}\circ \iota_{\widetilde{U}}(-) \circ v_{\widetilde{U}}$
is given by the adjoint action of a unitary isomorphism $v_{\widetilde{U}}\in 
\pi_{0W}\big(\AAA(W)\big)^{\prime\prime}\subseteq \pi_0(\AAA)_\star^{\prime\prime}$, 
which is weakly continuous and hence admits a unique weakly continuous extension $\rho_{\widetilde{U}} : 
\pi_{0\widetilde{U}}\big(\AAA(\widetilde{U})\big)^{\prime\prime}\to 
\pi_{0W}\big(\AAA(W)\big)^{\prime\prime} \subseteq \pi_0(\AAA)_\star^{\prime\prime}$
to the bicommutant. Our claim then follows
by passing to the colimit $\pi_0(\AAA)_\star^{\prime\prime}$.
\end{proof}

\begin{cor}\label{cor:endo}
Suppose that Assumption \ref{assu:monoidal} is satisfied. Then 
Lemma \ref{lem:toolsmonoidal} provides, for every $U\in\CC$, the following equivalent model for the 
$C^\ast$-category $\SSS_{(\AAA,\pi_0)}(U)$ from Definition \ref{def:SSSlocal}:
\begin{itemize}
\item[$\mathsf{Obj}$:] $\ast$-endomorphisms $\rho : \pi_0(\AAA)_\star^{\prime\prime}\to \pi_0(\AAA)_\star^{\prime\prime}$
satisfying the following properties:
\begin{itemize}
\item[(1)] For each $\widetilde{U}\in\CC$, there exists a unitary isomorphism
$v : H_0\to H_0$ such that
\begin{flalign}
v\circ \rho_{\widetilde{U}^\prime}(-)\,=\, \iota_{\widetilde{U}^\prime}(-)\circ v\quad,
\end{flalign}
for all orthogonal cospans $(\widetilde{U}^\prime\to V)\perp (\widetilde{U}\to V)$ in $\CC^\perp$,
where $\iota_{\widetilde{U}^\prime} : 
\pi_{0\widetilde{U}^\prime}\big(\AAA(\widetilde{U}^\prime)\big)^{\prime\prime}
\to  \pi_0(\AAA)_\star^{\prime\prime}$ denotes the canonical inclusion morphism into the colimit.

\item[(2)] $\rho_{U^\prime} = \iota_{U^\prime}$, for all orthogonal cospans 
$(U^\prime\to V)\perp (U\to V)$ relative to the fixed object $U\in\CC$.
\end{itemize}

\item[$\mathsf{Mor}$:] $T : \rho\to \rho^\prime$  are elements $T\in \pi_{0U}\big(\AAA(U)\big)^{\prime\prime}\subseteq \pi_0(\AAA)_\star^{\prime\prime}\subseteq \mathrm{B}(H_0)$
such that $T\circ\rho(-) = \rho^\prime(-)\circ T$. 
\end{itemize}
\end{cor}

This corollary is the key ingredient to recognize the following standard
monoidal structure on the $C^\ast$-categories of strictly localized superselection sectors, 
see also \cite{DHR,Roberts2,SSS}.
\begin{propo}\label{prop:monoidal}
Suppose that Assumption \ref{assu:monoidal} is satisfied. Then,
for every $U\in\CC$, the $C^\ast$-category $\SSS_{(\AAA,\pi_0)}(U)$ 
from Corollary \ref{cor:endo} can be endowed with a strict monoidal structure
given by the monoidal product $\ast$-functor
\begin{flalign}\label{eqn:monoidalproduct}
\cdiamond\,:\, \SSS_{(\AAA,\pi_0)}(U)\otimes \SSS_{(\AAA,\pi_0)}(U) ~&\longrightarrow~\SSS_{(\AAA,\pi_0)}(U)\quad,\\
\nn \rho\otimes \dot{\rho} ~&\longmapsto~ \rho\cdiamond \dot{\rho} := \rho\,\dot{\rho}\quad,\\
\nn (T:\rho\to \rho^\prime)\otimes(\dot{T} : \dot{\rho}\to \dot{\rho}^\prime)~&\longmapsto~
\big(T\cdiamond\dot{T} := T\circ \rho(\dot{T}) :\rho\cdiamond \dot{\rho}\to
\rho^\prime\cdiamond \dot{\rho}^\prime \big)\quad,
\end{flalign}
and the monoidal unit $\oone:= \id \in \SSS_{(\AAA,\pi_0)}(U)$. These object-wise monoidal structures
are compatible with the functor structure from Proposition \ref{prop:SSSfunctor} and thereby
define a locally constant functor
\begin{flalign}
\SSS_{(\AAA,\pi_0)}\,:\, \CC~\longrightarrow~\Alg_{\mathsf{uAs}}\big(\CastCat\big)
\end{flalign}
to the category of strict monoidal $C^\ast$-categories.
\end{propo}
\begin{proof}
Let us start with showing that \eqref{eqn:monoidalproduct} is well-defined,
i.e.\ $\rho\cdiamond \dot{\rho}$ satisfies the two properties (1) and (2) from Corollary \ref{cor:endo}.
Regarding (1), given any $\widetilde{U}\in\CC$, we have to find a unitary isomorphism $v^{\cdiamond} : H_0\to H_0$
such that $v^{\cdiamond}\circ (\rho\cdiamond \dot{\rho})_{\widetilde{U}^\prime}(-) =\iota_{\widetilde{U}^\prime}(-)
\circ v^{\cdiamond}$, for all orthogonal cospans $(\widetilde{U}^\prime \to V)\perp (\widetilde{U}\to V)$ in $\CC^\perp$.
Using that $\rho,\dot{\rho}\in \SSS_{(\AAA,\pi_0)}(U)$, there exist unitaries
$v,\dot{v}: H_0\to H_0$ such that $v\circ \rho_{\widetilde{U}^\prime}(-)=\iota_{\widetilde{U}^\prime}(-)\circ v$
and $\dot{v}\circ \dot{\rho}_{\widetilde{U}^\prime}(-)=\iota_{\widetilde{U}^\prime}(-)\circ \dot{v}$.
Using Lemma \ref{lem:toolsmonoidal} (i), it follows that $v,\dot{v} \in \pi_0(\AAA)_\star^{\prime\prime}$
and we can define the unitary isomorphism $v^{\cdiamond} := v\circ \rho(\dot{v}) : H_0\to H_0$,
which satisfies the desired property
\begin{flalign}
\nn v^{\cdiamond} \circ (\rho\cdiamond \dot{\rho})_{\widetilde{U}^\prime}(-)\,&=\,
v\circ \rho(\dot{v})\circ \rho\big(\dot{\rho}_{\widetilde{U}^\prime}(-)\big)\,=\,
v\circ \rho\big(\dot{v}\circ \dot{\rho}_{\widetilde{U}^\prime}(-)\big)\,=\,
v\circ \rho\big(\iota_{\widetilde{U}^\prime}(-)\circ \dot{v}\big)\\
\,&=\,v\circ \rho_{\widetilde{U}^\prime}(-)\circ \rho(\dot{v})\,=\,
\iota_{\widetilde{U}^\prime}(-)\circ v\circ \rho(\dot{v})\,=\, \iota_{\widetilde{U}^\prime}(-)\circ v^{\cdiamond}\quad.
\end{flalign}
To show property (2) from Corollary \ref{cor:endo}, note that given any orthogonal
cospan $(U^\prime\to V)\perp (U\to V)$ relative to the fixed object $U\in\CC$, we have
$(\rho\cdiamond \dot{\rho})_{U^\prime} = \rho\,\dot{\rho}_{U^\prime} = 
\rho\,\iota_{U^\prime} = \rho_{U^\prime} = \iota_{U^\prime}$, as required.
\sk

The assignment \eqref{eqn:monoidalproduct} is evidently a $\ast$-functor:
It preserves identity morphisms $\id_{\rho}\cdiamond\id_{\dot{\rho}} = 
\id_{\rho\diamond \dot{\rho}}$,
compositions, for $T:\rho\to\rho^\prime$, $T^\prime: \rho^\prime\to\rho^{\prime\prime}$,
$\dot{T}:\dot{\rho}\to\dot{\rho}^\prime$ and $\dot{T}^\prime: \dot{\rho}^\prime\to\dot{\rho}^{\prime\prime}$,
\begin{flalign}
\nn \big(T^\prime\circ T\big)\cdiamond\big(\dot{T}^\prime\circ \dot{T}\big)\,&=\,
T^\prime\circ T\circ \rho(\dot{T}^\prime)\circ \rho(\dot{T})\\
\,&=\,
T^\prime\circ \rho^\prime(\dot{T}^\prime)\circ T \circ \rho(\dot{T})\,=\,
\big(T^\prime\cdiamond \dot{T}^\prime\big)\circ\big(T\cdiamond \dot{T}\big)\quad,
\end{flalign}
and $\ast$-involutions $(T\cdiamond \dot{T})^\ast = 
\big(T\circ \rho(\dot{T})\big)^\ast = \big(\rho^\prime(\dot{T})\circ T\big)^\ast
= T^\ast\circ \rho^\prime(\dot{T}^\ast) = T^\ast\cdiamond \dot{T}^\ast$.
Strict associativity and unitality of the monoidal structures, as well as their compatibility
with the functor structure from Proposition \ref{prop:SSSfunctor}, are straightforward checks. 
\end{proof}

\paragraph{Prefactorization algebra structure:} 
The key observation which will lead to a 
prefactorization algebra structure on superselection sectors is the following: 
Under a simple geometric
assumption on the orthogonal category $\CC^\perp$, to be stated
in Assumption \ref{assu:PFA} below, the functor
$\SSS_{(\AAA,\pi_0)}:\CC\to \Alg_{\mathsf{uAs}}\big(\CastCat\big)$ from
Proposition \ref{prop:monoidal}, which assigns the strict monoidal $C^\ast$-categories
of strictly localized superselection sectors, satisfies a $\perp$-commutativity
property. In other words, it defines a categorified AQFT over
$\CC^\perp$ assigning strict monoidal $C^\ast$-categories, 
see also Example \ref{ex:PFA} (2) to compare with the more familiar concept 
of an AQFT over $\CC^\perp$ assigning $C^\ast$-algebras.
\begin{assu}\label{assu:PFA}
For every orthogonal cospan $(U_1\to \widetilde{U})\perp (U_2\to \widetilde{U})$ in $\CC^\perp$, 
there exists a $\CC$-morphism $\widetilde{U}\to W$ and orthogonal cospans
\begin{flalign}
\begin{gathered}
\xymatrix@C=0.6em@R=0.2em{
&&&W&&&\\
&&&\perp &&&\\
&&&&&&\\
&\ar[rruuu] V_1&&\widetilde{U} &&\ar[lluuu]V_2&\\
&\perp &&\perp &&\perp &\\
\ar[ruu]U_1^\prime && \ar[luu]U_1 \ar[ruu]&& \ar[luu]U_2 \ar[ruu]&& \ar[luu]U_2^\prime
}
\end{gathered}
\end{flalign}
in $\CC^\perp$, such that $(U_1^\prime\to V_1\to W)\perp (\widetilde{U}\to W)$
and $(U_2^\prime\to V_2\to W)\perp (\widetilde{U}\to W)$.
\end{assu}

\begin{propo}\label{prop:perpcommutativity}
Suppose that Assumptions \ref{assu:monoidal} and \ref{assu:PFA} are satisfied.
Then the functor $\SSS_{(\AAA,\pi_0)}:\CC\to \Alg_{\mathsf{uAs}}\big(\CastCat\big)$ from
Proposition \ref{prop:monoidal} satisfies the following $\perp$-commutativity property:
For every orthogonal cospan $(U_1\to V)\perp (U_2\to V)$ in $\CC^\perp$, the diagram
\begin{equation}\label{eqn:perpcommutativity}
\begin{tikzcd}
\SSS_{(\AAA,\pi_0)}(U_1)\otimes \SSS_{(\AAA,\pi_0)}(U_2)\ar[r]\ar[d] & 
\SSS_{(\AAA,\pi_0)}(V)\otimes \SSS_{(\AAA,\pi_0)}(V) \ar[d,"\diamond^\op"]\\[2mm]
\SSS_{(\AAA,\pi_0)}(V)\otimes \SSS_{(\AAA,\pi_0)}(V) \ar[r, "\diamond"'] & \SSS_{(\AAA,\pi_0)}(V)
\end{tikzcd}
\end{equation}
in $\CastCat$ commutes, where the unlabeled arrows are obtained from the functor structure of 
$\SSS_{(\AAA,\pi_0)}$ and $\cdiamond^{(\op)}$ denotes the (opposite) monoidal product.
\end{propo}
\begin{proof}
We have to verify commutativity of the diagram \eqref{eqn:perpcommutativity} at the level of objects
and morphisms. For two objects $\rho_1\in\SSS_{(\AAA,\pi_0)}(U_1)$
and $\rho_2\in\SSS_{(\AAA,\pi_0)}(U_2)$, we have to show that
\begin{flalign}
\rho_1\cdiamond \rho_2 \,=\, \rho_2\cdiamond \rho_1\quad\Longleftrightarrow\quad
(\rho_1\cdiamond \rho_2)\,\iota_{\widetilde{U}} \,=\, 
(\rho_2\cdiamond \rho_1)\,\iota_{\widetilde{U}}\quad \forall\, \widetilde{U}\in\CC\quad,
\end{flalign}
where the equivalence uses that the $\rho$'s are $\ast$-endomorphisms
on the universal $C^\ast$-algebra  $\pi_0(\AAA)_\star^{\prime\prime}$, which is defined as a colimit over $\CC$.
Using our assumption that $\CC$ is filtered, it suffices to consider
only those $\widetilde{U}\in\CC$ for which there exists a morphism $V\to \widetilde{U}$ in $\CC$.
From closedness of orthogonality relations under post-compositions,
it follows that $(U_1\to V\to \widetilde{U})\perp (U_2\to V\to \widetilde{U})$. Using
Assumption \ref{assu:PFA} for this orthogonal cospan and employing property (1) from Corollary \ref{cor:endo},
there exist unitaries $v_1\in \pi_{0V_1}\big(\AAA(V_1)\big)^{\prime\prime}$
and $v_2\in \pi_{0V_2}\big(\AAA(V_2)\big)^{\prime\prime}$ such that
\begin{flalign}
\rho_{1\widetilde{U}}\,=\,v_1^\ast\circ \iota_{\widetilde{U}}(-)\circ v_1\quad,\qquad
\rho_{2\widetilde{U}}\,=\,v_2^\ast\circ \iota_{\widetilde{U}}(-)\circ v_2\quad.
\end{flalign}
From this we compute
\begin{subequations}\label{eqn:tmpcommutativity}
\begin{flalign}
(\rho_1\cdiamond \rho_2)\,\iota_{\widetilde{U}}\,&=\,\rho_1\big(\rho_{2\widetilde{U}}(-)\big)\,=\,
\rho_1(v_2)^\ast\circ v_1^\ast \circ \iota_{\widetilde{U}}(-)\circ v_1\circ \rho_1(v_2)\quad,\\
(\rho_2\cdiamond \rho_1)\,\iota_{\widetilde{U}}\,&=\,\rho_2\big(\rho_{1\widetilde{U}}(-)\big)\,=\,
\rho_2(v_1)^\ast\circ v_2^\ast \circ \iota_{\widetilde{U}}(-)\circ v_2\circ \rho_2(v_1)\quad.
\end{flalign}
\end{subequations}
Using that $(V_2\to W)\perp (U_1\to V_1\to W)$ and $(V_1\to W)\perp (U_2\to V_2\to W)$,
property (2) from Corollary \ref{cor:endo} implies that $\rho_1(v_2)=v_2$ and $\rho_2(v_1)=v_1$,
hence the two expressions in \eqref{eqn:tmpcommutativity} are equal since $v_1\circ v_2=v_2\circ v_1$
as a consequence of $(V_1\to W)\perp (V_2\to W)$ in $\CC^\perp$.
\sk

It remains to prove that the diagram \eqref{eqn:perpcommutativity} commutes at the level of morphisms.
Given a morphism $T_1 :\rho_1\to \rho_1^\prime$ in $\SSS_{(\AAA,\pi_0)}(U_1)$ and
a morphism $T_2: \rho_2\to \rho_2^\prime$ in $\SSS_{(\AAA,\pi_0)}(U_2)$, 
it follows from Lemma \ref{lem:toolsmonoidal} (i) that $T_1\in \pi_{0U_1}\big(\AAA(U_1)\big)^{\prime\prime}$
and $T_2\in \pi_{0U_2}\big(\AAA(U_2)\big)^{\prime\prime}$, hence
$T_1\circ T_2=T_2\circ T_1$ since $(U_1\to V)\perp (U_2\to V)$ in $\CC^\perp$. 
From this we compute
\begin{flalign}
T_1\cdiamond T_2 \,=\, T_1\circ \rho_1(T_2) \,=\, T_1\circ T_2 \,=\, T_2\circ T_1\,=\, 
T_2\circ \rho_2(T_1)\,=\, T_2\cdiamond T_1\quad,
\end{flalign}
where steps two and four use property (2) from Corollary \ref{cor:endo}.
\end{proof}

Using this result, we can canonically extend the functor
$\SSS_{(\AAA,\pi_0)}:\CC\to \Alg_{\mathsf{uAs}}\big(\CastCat\big)$ 
from Proposition \ref{prop:monoidal} to a prefactorization algebra over $\CC^\perp$
whose factorization products, for all arity $n> 1$ operations $\und{f} : \und{U}\to V$ in $\P_{\CC^\perp}$, 
are defined by the $\ast$-functors
\begin{subequations}\label{eqn:PFAstructure}
\begin{equation}
\begin{tikzcd}
\bigotimes\limits_{i=1}^n \SSS_{(\AAA,\pi_0)}(U_i)\ar[rr,"\SSS_{(\AAA,\pi_0)}(\und{f})"]  
\ar[dr, "\Motimes_i \SSS_{(\AAA,\pi_0)}(f_i)"'] && \SSS_{(\AAA,\pi_0)}(V)\\
& \SSS_{(\AAA,\pi_0)}(V)^{\otimes n} \ar[ru,"\diamond^n"'] &
\end{tikzcd}
\quad,
\end{equation}
where $\cdiamond^n$ denotes the $n$-fold monoidal product. The factorization products
for the arity zero operations $()\to V$ are given by the monoidal units 
\begin{flalign}
\mathsf{B}\bbC~\longrightarrow~\SSS_{(\AAA,\pi_0)}(V)\quad.
\end{flalign}
\end{subequations}
This leads us to the main result of this section.
\begin{theo}\label{theo:PFAstructure}
Suppose that Assumptions \ref{assu:monoidal} and \ref{assu:PFA} are satisfied.
Then the $\ast$-functors in \eqref{eqn:PFAstructure} are strictly monoidal and they 
endow the functor $\SSS_{(\AAA,\pi_0)}:\CC\to \Alg_{\mathsf{uAs}}\big(\CastCat\big)$ 
from Proposition \ref{prop:monoidal} with the structure of a prefactorization algebra
\begin{flalign}
\SSS_{(\AAA,\pi_0)}\,\in\,\Alg_{\P_{\CC^\perp}}^{}\!\Big(\Alg_{\mathsf{uAs}}\big(\CastCat\big)\Big)
\end{flalign}
taking values in the category of strict monoidal $C^\ast$-categories. This prefactorization
algebra is locally constant in the sense that it assigns
to every $1$-ary operation $f:U\to \widetilde{U}$ in $\P_{\CC^\perp}$
a unitary equivalence of strict monoidal $C^\ast$-categories.
\end{theo}
\begin{proof}
To show that the $\ast$-functors \eqref{eqn:PFAstructure} define 
a $\P_{\CC^\perp}$-algebra structure, one has to verify that they 
preserve operadic compositions, identity operations and permutation actions,
see also Remark \ref{rem:PFA}.
The first two properties are evident from the definition of the structure maps
in \eqref{eqn:PFAstructure}, while permutation equivariance follows 
by applying Proposition \ref{prop:perpcommutativity} to the orthogonal 
cospans $(f_i : U_i\to V)\perp (f_j : U_j\to V)$, for all $i\neq j$.
\sk

Since local constancy has already been observed in Proposition \ref{prop:monoidal},
it remains to prove that the $\ast$-functors \eqref{eqn:PFAstructure} are strictly monoidal,
which requires checks at the level of objects and at the level of morphisms.
Since the functor structure $\SSS_{(\AAA,\pi_0)}(f_i):\SSS_{(\AAA,\pi_0)}(U_i)\to\SSS_{(\AAA,\pi_0)}(V)$
is simply given by a full $C^\ast$-subcategory inclusion $\SSS_{(\AAA,\pi_0)}(U_i)\subseteq \SSS_{(\AAA,\pi_0)}(V)$,
we can and will suppress the symbols $\SSS_{(\AAA,\pi_0)}(f_i)$ from these checks.
Given any objects $\rho_i,\dot{\rho}_i\in \SSS_{(\AAA,\pi_0)}(U_i)$, for all $i=1,\dots,n$,
we have to show that
\begin{subequations}
\begin{flalign}
\SSS_{(\AAA,\pi_0)}(\und{f})\big((\rho_1\cdiamond\dot{\rho}_1)\otimes \cdots \otimes(\rho_n\cdiamond\dot{\rho}_n)\big)
\,=\, \rho_1\cdiamond\dot{\rho}_1\cdiamond\cdots\cdiamond \rho_n\cdiamond\dot{\rho}_n
\end{flalign}
is equal to
\begin{flalign}
\SSS_{(\AAA,\pi_0)}(\und{f})\big(\rho_1\otimes\cdots\otimes \rho_n\big)\cdiamond
\SSS_{(\AAA,\pi_0)}(\und{f})\big(\dot{\rho}_1\otimes\cdots\otimes \dot{\rho}_n\big)\,=\, 
\rho_1\cdiamond\cdots\cdiamond \rho_n\cdiamond \dot{\rho}_1\cdiamond \cdots \cdiamond \dot{\rho}_n\quad.
\end{flalign}
\end{subequations}
This follows by applying the $\perp$-commutativity property (for objects) from
Proposition \ref{prop:perpcommutativity} to the orthogonal 
cospans $(f_i : U_i\to V)\perp (f_j : U_j\to V)$, for all $i\neq j$.
A similar check can be performed at the level of morphisms. 
\end{proof}


\section{\label{sec:dEqvSSS}Equivariance under discrete group actions}
The aim of this section is to extend the results from Section \ref{sec:SSS}
to a $G$-equivariant context, for $G$ a discrete group. 
Throughout this section we fix an orthogonal category with $G$-action
$\CC^\perp : \mathsf{B}G\to \Cat^\perp$ and assume as before
that its underlying category $\CC$ is filtered.
The starting point of our construction is a $G$-equivariant AQFT $\big(\AAA,\Psi\big)\in \AQFT^G\big(\CC^\perp\big)$ 
in the sense of Definition \ref{def:eqvAQFT}, which explicitly consists of
an AQFT $\AAA\in\AQFT\big(\CC^\perp\big)$ over the underlying orthogonal category 
and a family of $\AQFT\big(\CC^\perp\big)$-isomorphisms 
\begin{flalign}
\Psi\,=\,\Big\{\Psi_g \,:\, \AAA\,\stackrel{\cong}{\longrightarrow}\, \alpha_g^\ast(\AAA)\Big\}_{g\in G}
\end{flalign}
satisfying the properties spelled out in Remark \ref{rem:eqvPFA}.
Our goal is to induce this equivariance structure from the AQFT $\AAA$ 
to the prefactorization algebra of superselection sectors 
from Theorem \ref{theo:PFAstructure}, thereby endowing it with the structure 
of a $G$-equivariant prefactorization algebra.
\sk

As a first step, let us observe that the $G$-equivariance structure $\Psi$ 
on the AQFT $\AAA$ induces a $G$-action $\big\{\Psi_{\star g}:\AAA_\star\to\AAA_\star\big\}_{g\in G}$
on its universal $C^\ast$-algebra \eqref{eqn:universal}, which is given 
explicitly by
\begin{equation}
\begin{tikzcd}
\AAA_\star \ar[r,dashed,"\Psi_{\star g}"] & \AAA_\star\\[2mm]
\ar[u,"\iota_U"] \AAA(U) \ar[r, "(\Psi_g)_U^{}"'] & \AAA(\alpha_g(U))\ar[u,"\iota_{\alpha_g(U)}"']
\end{tikzcd}
\quad,
\end{equation}
for all $U\in\CC$ and all $g\in G$. From this one obtains via pre-composition a (right) $G$-action 
\begin{flalign}\label{eqn:GactionRep}
\Psi_{\star g}^\ast \,:\, \Rep_{\AAA}~&\longrightarrow~\Rep_{\AAA}\quad,\\
\nn \big(\pi: \AAA_\star \to \mathrm{B}(H)\big)~&\longmapsto~\big(\pi\,\Psi_{\star g} : \AAA_\star \to \mathrm{B}(H)\big)\quad,\\
\nn \big(T: \pi\to \pi^\prime\big)~&\longmapsto~\big(T: \pi\,\Psi_{\star g}\to \pi^\prime\,\Psi_{\star g} \big)\quad,
\end{flalign}
for all $g\in G$, on the representation $C^\ast$-category $\Rep_{\AAA}$.
To restrict this $G$-action to the full $C^\ast$-subcategory $\SSS_{(\AAA,\pi_0)}\subseteq \Rep_\AAA$ 
of superselection sectors from Definition \ref{def:SSSglobal} we impose the following assumption.
\begin{assu}\label{assu:Grep}
There exists a projective unitary representation $u : G\to \mathrm{PU}(H_0)\,,\,g\mapsto u_g$ on the Hilbert space
of the reference representation $\pi_0 : \AAA_\star\to \mathrm{B}(H_0)$ such that
\begin{flalign}
\pi_0\,\Psi_{\star g} \,=\, \Ad_{u_g}\,\pi_0\quad,
\end{flalign}
for all $g\in G$, where $\Ad_{u_g}:\mathrm{B}(H_0)\to \mathrm{B}(H_0)\,,~T\mapsto u_g\circ T\circ u_g^\ast$ denotes the adjoint action.
\end{assu}

\begin{lem}\label{lem:GactionSSS}
Under Assumption \ref{assu:Grep}, the (right) $G$-action \eqref{eqn:GactionRep} on the 
representation $C^\ast$-category $\Rep_{\AAA}$ restricts to a (right) $G$-action 
on the full $C^\ast$-subcategory $\SSS_{(\AAA,\pi_0)}\subseteq \Rep_{\AAA}$ of superselection sectors
from Definition \ref{def:SSSglobal}.
\end{lem}
\begin{proof}
We have to show that, for every $g\in G$ and $\big(\pi:\AAA_\star\to \mathrm{B}(H)\big)\in \SSS_{(\AAA,\pi_0)}$, the representation
$\Psi_{\star g}^\ast(\pi) = \pi \,\Psi_{\star g}\in\Rep_{\AAA}$ satisfies the localizability conditions from 
Definition \ref{def:SSSglobal}. In more detail, given any $U\in \CC$, we have to construct a unitary isomorphism
$v : H\to H_0$ such that
\begin{flalign}\label{eqn:tmpG0}
v\circ \Psi_{\star g}^\ast(\pi)_{U^\prime}(-)\,=\,\pi_{0 U^\prime}(-)\circ v\quad,
\end{flalign}
for all orthogonal cospans $(U^\prime\to V)\perp (U\to V)$ in $\CC^\perp$. 
Using the definitions above, one can express 
the term $\Psi_{\star g}^\ast(\pi)_{U^\prime}$ on the left-hand side as
\begin{subequations}\label{eqn:tmpG1}
\begin{flalign}
\Psi_{\star g}^\ast(\pi)_{U^\prime} \,=\, \pi\,\Psi_{\star g}\,\iota_{U^\prime}\,=\,
\pi\, \iota_{\alpha_g(U^\prime)}\, (\Psi_{g})_{U^\prime}^{}\,=\, \pi_{\alpha_g(U^\prime)}\,(\Psi_{g})_{U^\prime}^{}\quad.
\end{flalign}
Using also Assumption \ref{assu:Grep}, we further have that
\begin{flalign}
\Psi_{\star g}^\ast(\pi_0)_{U^\prime}\,=\,\pi_{0 \alpha_g(U^\prime)}\,(\Psi_{g})_{U^\prime}^{}\,=\,
\Ad_{u_g}\,\pi_{0 U^\prime}\quad\Longleftrightarrow\quad
\pi_{0 U^\prime}\,=\, \Ad_{u_{g}^\ast}\,\pi_{0 \alpha_g(U^\prime)}\,(\Psi_{g})_{U^\prime}^{}\quad.
\end{flalign}
\end{subequations}
Inserting \eqref{eqn:tmpG1} into \eqref{eqn:tmpG0}, one arrives at the equivalent requirement that
\begin{flalign}
u_{g}\circ v\circ \pi_{\alpha_g(U^\prime)}\big((\Psi_{g})_{U^\prime}^{}(-)\big)\,=\,
\pi_{0 \alpha_g(U^\prime)}\big((\Psi_{g})_{U^\prime}^{}(-)\big)\circ u_g\circ v\quad, 
\end{flalign}
for all orthogonal cospans $(U^\prime\to V)\perp (U\to V)$ in $\CC^\perp$. 
Picking any unitary isomorphism $\widetilde{v} : H\to H_0$ witnessing the 
localizability condition for $\pi\in \SSS_{(\AAA,\pi_0)}$ in $\alpha_g(U)\in \CC$, 
we can set $v := u_{g}^\ast\circ \widetilde{v}$ to obtain a unitary isomorphism with the desired properties,
hence $\Psi_{\star g}^\ast(\pi) \in\SSS_{(\AAA,\pi_0)}$.
\end{proof}

The $G$-action on superselection sectors established in this lemma 
has the deficit that it \textit{does not} preserve the distinguished object 
$\pi_0\in \SSS_{(\AAA,\pi_0)}$ and also \textit{does not} interplay well 
with the strict localization conditions from Definition \ref{def:SSSlocal}.
Observing that Assumption \ref{assu:Grep} can be rewritten
as the statement that the reference representation $\pi_0$ is invariant
under the combined (left) $G$-action $\Ad_{u_g}\,\pi_0\,\Psi_{\star g^{-1}} = \pi_0$,
for all $g\in G$, leads us to the following improvement.
\begin{propo}\label{prop:GactionSSSimproved}
Suppose that Assumption \ref{assu:Grep} is satisfied and 
choose a projective unitary representation $u:G\to \mathrm{PU}(H_0)$.
Denote by $\SSS_{(\AAA,\pi_0)}\vert_{H_0}\subseteq \SSS_{(\AAA,\pi_0)}$ the unitarily equivalent 
full $C^\ast$-subcategory consisting of all superselection sectors from Definition \ref{def:SSSglobal} which
share the same Hilbert space $H_0$ with the reference representation $\pi_0$.
Then the family of $\ast$-functors
\begin{flalign}\label{eqn:GactionSSSimproved}
g\triangleright(-)\,:\,\SSS_{(\AAA,\pi_0)}\vert_{H_0}~&\longrightarrow~\SSS_{(\AAA,\pi_0)}\vert_{H_0}\quad,\\
\nn \pi~&\longmapsto~\Ad_{u_g}\,\pi\,\Psi_{\star g^{-1}}\quad,\\
\nn \big(T:\pi\to \pi^\prime\big)~&\longmapsto~\big(\Ad_{u_g}(T):\Ad_{u_g}\,\pi\,\Psi_{\star g^{-1}}\to 
\Ad_{u_g}\,\pi^\prime\,\Psi_{\star g^{-1}}\big)\quad,
\end{flalign}
for all $g\in G$, defines a (left) $G$-action.
This $G$-action leaves invariant the distinguished object, i.e.\ $g\triangleright \pi_0 = \pi_0$
for all $g\in G$, and it is compatible with the strict localization conditions from 
Definition \ref{def:SSSlocal} in the sense that
\begin{flalign}\label{eqn:GactionSSSimprovedlocal}
g\triangleright(-)\,:\,\SSS_{(\AAA,\pi_0)}(U)~\longrightarrow~\SSS_{(\AAA,\pi_0)}(\alpha_g(U))\quad,
\end{flalign}
for all $g\in G$ and $U\in \CC$.
\end{propo}
\begin{proof}
Let us start with observing that the $\ast$-functors \eqref{eqn:GactionSSSimproved} 
are indeed well-defined: We have shown in Lemma \ref{lem:GactionSSS} that 
$\pi\,\Psi_{\star g^{-1}}\in \SSS_{(\AAA,\pi_0)}\vert_{H_0}$,
for all $\pi\in \SSS_{(\AAA,\pi_0)}\vert_{H_0}$, and post-composition with $\Ad_{u_g}$ clearly
preserves the localizability conditions for superselection sectors from Definition \ref{def:SSSglobal}.
Invariance of the distinguished object $\pi_0\in \SSS_{(\AAA,\pi_0)}\vert_{H_0}$ is a direct consequence
of Assumption \ref{assu:Grep}. It remains to prove the claim about compatibility with the strict localization conditions.
Given any $U\in \CC$ and $\pi\in \SSS_{(\AAA,\pi_0)}(U)\subseteq \SSS_{(\AAA,\pi_0)}\vert_{H_0}$, 
we use that for every orthogonal cospan  $(U^\prime\to V)\perp (\alpha_g(U)\to V)$ in $\CC^\perp$
the cospan $(\alpha_{g^{-1}}(U^\prime) \to \alpha_{g^{-1}}(V))\perp (U\to \alpha_{g^{-1}}(V))$ is orthogonal too, 
which allows us to compute
\begin{flalign}
\nn (g\triangleright \pi)_{U^\prime}(-)\,&=\,
u_g\circ \pi_{\alpha_{g^{-1}}(U^\prime)}\big((\Psi_{g^{-1}})_{U^\prime}^{}(-)\big)\circ u_g^\ast\\
\nn \,&=\,u_g\circ \pi_{0\alpha_{g^{-1}}(U^\prime)}\big((\Psi_{g^{-1}})_{U^\prime}^{}(-)\big)\circ u_g^\ast\\
\,&=\,(g\triangleright \pi_0)_{U^\prime}(-)\,=\,\pi_{0U^\prime}(-)\quad,
\end{flalign}
where in the second step we used the strict $U$-localization condition for $\pi\in \SSS_{(\AAA,\pi_0)}(U)$
and in the last step we used that $\pi_0$ is invariant.
This implies that  $g\triangleright \pi \in \SSS_{(\AAA,\pi_0)}(\alpha_g(U))$ is strictly localized in $\alpha_g(U)\in \CC$.
\end{proof}

Since \eqref{eqn:GactionSSSimprovedlocal} is defined as the 
restriction and corestriction of a $G$-action on the full $C^\ast$-subcategory 
$\SSS_{(\AAA,\pi_0)}\vert_{H_0}\subseteq \SSS_{(\AAA,\pi_0)}$
of all superselection sectors with Hilbert space $H_0$, it follows directly that the diagram of $\ast$-functors
\begin{equation}\label{eqn:Gactionfunctor}
\begin{tikzcd}
\SSS_{(\AAA,\pi_0)}(U) \ar[d, "g\,\triangleright(-)"'] \ar[r, hookrightarrow] & 
\SSS_{(\AAA,\pi_0)}(\widetilde{U})\ar[d, "g\,\triangleright(-)"]\\
\SSS_{(\AAA,\pi_0)}(\alpha_g(U)) \ar[r, hookrightarrow] &  \SSS_{(\AAA,\pi_0)}(\alpha_g(\widetilde{U}))
\end{tikzcd}
\end{equation}
commutes, for all $g\in G$ and all morphisms $f:U\to \widetilde{U}$ in $\CC$.
This implies that $\triangleright$ endows the locally constant functor $\SSS_{(\AAA,\pi_0)}:\CC\to \CastCat$
from Proposition \ref{prop:SSSfunctor} with a $G$-equivariance structure.
To understand the interplay between this $G$-equivariance structure and the monoidal
structures from Proposition \ref{prop:monoidal}, we have to pass via Corollary \ref{cor:endo} 
to the equivalent models for the $C^\ast$-categories $\SSS_{(\AAA,\pi_0)}(U)$ of 
strictly $U$-localized superselection sectors which are given in terms of $\ast$-endomorphisms. 
In these models the $G$-equivariance structure from Proposition \ref{prop:GactionSSSimproved}
reads as
\begin{flalign}\label{eqn:GactionSSSendo}
g\triangleright(-)\,:\,\SSS_{(\AAA,\pi_0)}(U)~&\longrightarrow~\SSS_{(\AAA,\pi_0)}(\alpha_g(U))\quad,\\
\nn \big(\rho: \pi_0(\AAA)_\star^{\prime\prime}\to \pi_0(\AAA)_\star^{\prime\prime}\big)~&
\longmapsto~\big(\Ad_{u_g}\,\rho\,\Ad_{u_g^\ast}: \pi_0(\AAA)_\star^{\prime\prime}\to \pi_0(\AAA)_\star^{\prime\prime}\big)\quad,\\
\nn \big(T:\rho\to \rho^\prime\big)~&\longmapsto~\big(\Ad_{u_g}(T):\Ad_{u_g}\,\rho\,\Ad_{u_g^\ast}\to 
\Ad_{u_g}\,\rho^\prime\,\Ad_{u_g^\ast}\big)\quad,
\end{flalign}
for all $g\in G$ and $U\in\CC$.
\begin{propo}\label{prop:Gactionmonoidal}
Suppose that Assumptions \ref{assu:monoidal} and \ref{assu:Grep} are satisfied.
Then the $G$-equivariance structure $\triangleright$ from \eqref{eqn:GactionSSSendo}
is compatible with the monoidal structures $\cdiamond$ from Proposition \ref{prop:monoidal}
in the sense that the diagram of $\ast$-functors
\begin{equation}\label{eqn:Gactionmonoidal}
\begin{tikzcd}
\SSS_{(\AAA,\pi_0)}(U)\otimes\SSS_{(\AAA,\pi_0)}(U) \ar[r, "\cdiamond"] 
\ar[d,"g\,\triangleright(-)\otimes g\,\triangleright(-)"'] &[10mm] \SSS_{(\AAA,\pi_0)}(U)\ar[d,"g\,\triangleright(-)"]\\
\SSS_{(\AAA,\pi_0)}(\alpha_g(U))\otimes \SSS_{(\AAA,\pi_0)}(\alpha_g(U)) \ar[r, "\diamond"'] & \SSS_{(\AAA,\pi_0)}(\alpha_g(U))
\end{tikzcd}
\end{equation}
commutes, for all $g\in G$ and $U\in \CC$.
This implies that $\triangleright$ endows the locally constant functor 
$\SSS_{(\AAA,\pi_0)}:\CC\to \Alg_{\mathsf{uAs}}\big(\CastCat\big)$
from Proposition \ref{prop:monoidal} with a $G$-equivariance structure.
\end{propo}
\begin{proof}
Verifying commutativity of the diagram \eqref{eqn:Gactionmonoidal} is an elementary check,
which for completeness we shall briefly spell out. At the
level of objects $\rho,\dot{\rho}\in \SSS_{(\AAA,\pi_0)}(U)$, commutativity 
is shown by
\begin{flalign}
(g\triangleright \rho)\cdiamond (g\triangleright\dot{\rho}) \,=\, 
\Ad_{u_g}\,\rho\,\Ad_{u_g^\ast} \,\Ad_{u_g}\,\dot{\rho}\,\Ad_{u_g^\ast}\,=\,
\Ad_{u_g}\,\rho\,\dot{\rho}\,\Ad_{u_g^\ast} \,=\,g\triangleright(\rho\cdiamond\dot{\rho})\quad.
\end{flalign}
At the level of morphisms $T:\rho\to \rho^\prime$ and $\dot{T}:\dot{\rho}\to \dot{\rho}^\prime$ in $\SSS_{(\AAA,\pi_0)}(U)$,
commutativity is shown by
\begin{flalign}
(g\triangleright T)\cdiamond (g\triangleright\dot{T})\,=\,
u_g\circ T\circ u_g^{\ast}\circ (g\triangleright\rho)\big(u_g\circ \dot{T}\circ u_g^{\ast}\big)
\,=\,u_g\circ T\circ \rho(\dot{T})\circ u_g^{\ast} \,=\,g\triangleright(T\cdiamond\dot{T})\quad,
\end{flalign}
where the second step follows by inserting $g\triangleright\rho = \Ad_{u_g}\,\rho\,\Ad_{u_g^\ast}$.
\end{proof}

Combining the results above leads us to one of the main results of this section.
\begin{theo}\label{theo:SSSeqvPFA}
Suppose that Assumptions \ref{assu:monoidal}, \ref{assu:PFA} and \ref{assu:Grep} are satisfied.
Then $\triangleright$ given in \eqref{eqn:GactionSSSendo} endows
the locally constant prefactorization algebra  
from Theorem \ref{theo:PFAstructure} with the structure of a $G$-equivariant locally constant
prefactorization algebra $\SSS_{(\AAA,\pi_0)}\in\Alg^G_{\P_{\CC^\perp}}\!\big(\Alg_{\mathsf{uAs}}\big(\CastCat\big)\big)$.
\end{theo}
\begin{proof}
Compatibility of $\triangleright$ with the prefactorization algebra structure
\eqref{eqn:PFAstructure} follows directly from \eqref{eqn:Gactionfunctor},
\eqref{eqn:Gactionmonoidal} and the fact that $g\triangleright\oone= \oone$ 
leaves invariant the monoidal units.
\end{proof}

In the literature discussing group actions on superselection sectors, one often
finds the following notion of covariant superselection sectors, see e.g.\ 
\cite[Section 4.1.3]{Roberts1}, \cite[Section 8.4]{Halvorson} and \cite[Definition 3.7]{CRV}.
\begin{defi}\label{def:covariantSSS}
The $C^\ast$-category of \textit{covariant superselection sectors} 
strictly localized in $U\in \CC$ is defined as the full $C^\ast$-subcategory
\begin{flalign}
\SSS^{\mathrm{cov}}_{(\AAA,\pi_0)}(U)\,\subseteq\,\SSS_{(\AAA,\pi_0)}(U)
\end{flalign}
consisting of all objects $\rho\in \SSS_{(\AAA,\pi_0)}(U)$ for which there
exists a projective unitary representation $u^\rho : G\to \mathrm{PU}(H_0)\,,\,g\mapsto u^\rho_g$
such that $\rho \,\Ad_{u_g} \,=\, \Ad_{u^\rho_g}\,\rho$, for all $g\in G$. 
\end{defi}

\begin{rem}\label{rem:covSSS}
Given any $U\in\CC$ and $\rho\in \SSS^\mathrm{cov}_{(\AAA,\pi_0)}(U)$,
one can write
\begin{flalign}
g\triangleright\rho\,=\, \Ad_{u_g}\,\rho\,\Ad_{u_g^\ast} \,=\, \Ad_{u_g\circ u^{\rho\ast}_g}\,\rho\quad,
\end{flalign}
for all $g\in G$. From Lemma \ref{lem:toolsmonoidal} (i) together with
$\rho\in \SSS_{(\AAA,\pi_0)}(U)$ and $g\triangleright \rho\in \SSS_{(\AAA,\pi_0)}(\alpha_g(U))$,
it follows that
\begin{flalign}
u_g \circ u^{\rho\ast}_g\,\in\, \pi_0(\AAA)^{\prime\prime}_\star\big/(\bbC^\times \oone)
\end{flalign}
defines an element of the universal algebra modulo non-zero scalar multiples of the unit, 
for all $g\in G$. These elements will be important for our constructions below.
\end{rem}

Our goal in the remaining part of this section is to show that 
the $G$-equivariant locally constant prefactorization algebra 
structure $\SSS_{(\AAA,\pi_0)}\in \Alg^G_{\P_{\CC^\perp}}\!\big(\Alg_{\mathsf{uAs}}\big(\CastCat\big)\big)$ 
from Theorem \ref{theo:SSSeqvPFA} restricts to the covariant superselection sectors
from Definition \ref{def:covariantSSS}. Let us start with discussing 
the $G$-action \eqref{eqn:GactionSSSendo}. 
Given any $g\in G$, $U\in\CC$ and $\rho \in \SSS^{\mathrm{cov}}_{(\AAA,\pi_0)}(U)$,
we choose any $u^{\rho}: G\to \mathrm{PU}(H_0)$ as in Definition \ref{def:covariantSSS} 
and define the projective unitary representation
\begin{flalign}
u^{g\,\triangleright \rho}\,:\,G\, \longrightarrow~\mathrm{PU}(H_0)\,,~
g^\prime\,\longmapsto\,u^{g\,\triangleright \rho}_{g^\prime}\,:=\,
\Ad_{u_g\circ u^{\rho\ast}_g}(u^{\rho}_{g^\prime})\,=\,
(u_g \circ u^{\rho\ast}_g)\circ u^{\rho}_{g^\prime}\circ( u^\rho_{g} \circ u_g^\ast)\quad.
\end{flalign}
One immediately verifies that $(g\triangleright\rho)\,\Ad_{u_{g^\prime}} =
\Ad_{u^{g\,\triangleright \rho}_{g^\prime}}\,(g\triangleright\rho)$, for all $g^\prime\in G$,
hence we obtain a restriction
\begin{flalign}\label{eqn:GactionSSScov}
g\triangleright(-)\,:\,\SSS^\mathrm{cov}_{(\AAA,\pi_0)}(U)~\longrightarrow~\SSS^\mathrm{cov}_{(\AAA,\pi_0)}(\alpha_g(U))
\end{flalign}
of the $G$-action \eqref{eqn:GactionSSSendo} to covariant superselection sectors.
Let us discuss now the monoidal structures $\cdiamond$ from Proposition \ref{prop:monoidal}.
Given any $U\in\CC$ and $\rho,\dot{\rho}\in \SSS^{\mathrm{cov}}_{(\AAA,\pi_0)}(U)$,
we choose any $u^{\rho},u^{\dot{\rho}}: G\to \mathrm{PU}(H_0)$ as in Definition \ref{def:covariantSSS} 
and define the projective unitary representation
\begin{flalign}
u^{\rho\cdiamond \dot{\rho}}\,:\, G\,\longrightarrow\, \mathrm{PU}(H_0)\,,~g\,\longmapsto\, 
u^{\rho\cdiamond \dot{\rho}}_g\,:=\,u^\rho_g\circ \rho(u_g^\ast\circ u^{\dot{\rho}}_g) \,=\,
\rho(u^{\dot{\rho}}_g \circ u_g^\ast)\circ  u^\rho_g
\quad,
\end{flalign}
where the term involving $\rho$ is well-defined 
since $u^{\dot{\rho}}_g \circ u_g^\ast\in \pi_0(\AAA)^{\prime\prime}_\star/(\bbC^\times\oone)$
by Remark \ref{rem:covSSS}. 
One immediately verifies that 
\begin{flalign}
\nn (\rho\cdiamond \dot{\rho})\,\Ad_{u_g}\,&=\,
\rho\,\dot{\rho}\,\Ad_{u_g}\,=\,\rho\,\Ad_{u^{\dot{\rho}}_g}\,\dot{\rho}\,=\,
\rho\,\Ad_{u_g}\,\Ad_{u_g^\ast \circ u^{\dot{\rho}}_g}\,\dot{\rho}\\
\,&=\,
\Ad_{u^{\rho}_g}\, \rho\,\Ad_{u_g^\ast \circ u^{\dot{\rho}}_g}\,\dot{\rho}\,=\,
\Ad_{u^{\rho}_g \circ \rho(u_g^\ast\circ u^{\dot{\rho}}_g)}\,\rho\,\dot{\rho} \,=\, 
\Ad_{u_g^{\rho\cdiamond \dot{\rho}}}\,(\rho\cdiamond\dot{\rho})\quad,
\end{flalign}
for all $g\in G$, hence we obtain a restriction
\begin{flalign}\label{eqn:monoidalSSScov}
\cdiamond\,:\, \SSS^{\mathrm{cov}}_{(\AAA,\pi_0)}(U)\otimes \SSS^{\mathrm{cov}}_{(\AAA,\pi_0)}(U) 
~\longrightarrow~\SSS^{\mathrm{cov}}_{(\AAA,\pi_0)}(U)
\end{flalign}
of the monoidal structures $\cdiamond$ from Proposition \ref{prop:monoidal}
to covariant superselection sectors. Note that the monoidal unit objects
$\oone=\id\in \SSS^{\mathrm{cov}}_{(\AAA,\pi_0)}(U)$ are covariant with respect
to the reference projective unitary representation $u^{\oone}=u: G\to \mathrm{PU}(H_0)$
from Assumption \ref{assu:Grep}. This leads us to the following result.
\begin{theo}\label{theo:covSSSeqvPFA}
Suppose that Assumptions \ref{assu:monoidal}, \ref{assu:PFA} and \ref{assu:Grep} are satisfied.
Then the $G$-equivariant locally constant prefactorization algebra structure 
from Theorem \ref{theo:SSSeqvPFA} restricts to a $G$-equivariant locally constant prefactorization algebra structure
$\SSS_{(\AAA,\pi_0)}^{\mathrm{cov}}\in\Alg^G_{\P_{\CC^\perp}}\!\big(\Alg_{\mathsf{uAs}}\big(\CastCat\big)\big)$
on the covariant superselection sectors from Definition \ref{def:covariantSSS}.
\end{theo}


\section{\label{sec:examples}Example: Double cones in Minkowski spacetime}
In this section we apply our framework to one 
of the prime examples of superselection theory,
given by AQFTs which are defined on double cones in the $(n\geq 2)$-dimensional
Minkowski spacetime $\bbM^n := \bbR^1\times\Sigma^{n-1}:= \bbR^1\times \bbR^{n-1}$, with metric
convention $g= -\dd t^2 + \dd\mathbf{x}^2$. The main result of this section
is a conceptually clean and elegant proof that, in this
particularly simple geometric context, our locally constant 
prefactorization algebras of superselection sectors 
from Theorem \ref{theo:PFAstructure} admit an equivalent
yet simpler description in terms of $\bbE_n$-monoidal $C^\ast$-categories.
Our proof unravels and explains that this well-known $\bbE_n$-monoidal structure
emerges, through Dunn-Lurie additivity $\bbE_n\simeq \bbE_1\otimes \bbE_{n-1}$,
from an $\bbE_1$-monoidal structure of analytical origin in Haag duality
and an $\bbE_{n-1}$-monoidal structure of geometric origin 
in the $(n-1)$-dimensional spatial factor $\Sigma^{n-1}$.
We conclude with a brief observation about $G$-equivariant
generalizations of this result.
\begin{defi}\label{def:doublecone}
A \textit{double cone} in the $(n\geq 2)$-dimensional Minkowski spacetime $\bbM^n$
is a non-empty open subset of the form 
\begin{flalign}
U \,=\, I^-(\{p^+\})\cap I^+(\{p^-\})\,\subseteq\, \bbM^n\quad,
\end{flalign}
where $p^\pm\in\bbM^n$ are two points and $I^\pm(\{p^\mp\})\subseteq \bbM^n$ 
denote their chronological future/past.\footnote{Note that $U\neq \emptyset$
if and only if $p^+ \in I^{+}(\{p^-\})$. Geometrically, $p^+$ is the future 
tip of the double cone and $p^-$ is its past tip.}
We denote by $\mathbf{DCone}_{\bbM^n}^\perp$ the orthogonal category
whose objects are all double cones $U\subseteq \bbM^n$, morphisms are all
subset inclusions $U\subseteq U^\prime$, and orthogonality relation is given by causal disjointness,
i.e.\ $(U_1\subseteq V)\perp(U_2\subseteq V)$ if and only if $J(U_1)\cap U_2=\varnothing$,
where $J(U_1):= J^+(U_1)\cup J^-(U_1)\subseteq \bbM^n$ denotes the union of
the causal future and past of $U_1\subseteq \bbM^n$.
\end{defi}

\begin{rem}\label{rem:doublecone}
To simplify geometric arguments about double cones, we will frequently make use of the canonical
affine space structure $\bbM^n\times \bbM^n\to \bbR^n\,,~(p,p^\prime)\mapsto p-p^\prime$ 
on the Minkowski spacetime $\bbM^n$, with corresponding free and transitive
translation action denoted by $\bbM^n\times \bbR^n\to \bbM^n\,,~(p,v)\mapsto p+v$. With respect to this
affine space structure, every double cone $U = I^-(\{p^+\})\cap I^+(\{p^-\})\subseteq \bbM^n$
is a convex open subset, i.e.\ $p^\prime + s(p-p^\prime)\in U$ for all $p,p^\prime\in U$ and $s\in[0,1]$,
with distinguished center point $p^\circ := p^- + \frac{1}{2}(p^+-p^-)\in U$.
\end{rem}

\begin{lem}\label{lem:doublecone}
The orthogonal category $\mathbf{DCone}_{\bbM^n}^\perp$ of double cones from Definition \ref{def:doublecone}
is filtered and satisfies the geometric properties from Assumptions \ref{assu:monoidal} (1) and \ref{assu:PFA}.
\end{lem}
\begin{proof}
Filteredness and the geometric property from Assumption \ref{assu:monoidal} (1) are obvious,
while the one from Assumption \ref{assu:PFA} requires a brief argument.
Let $(U_1\subseteq \widetilde{U})\perp (U_2\subseteq\widetilde{U})$ be any orthogonal cospan
in $\mathbf{DCone}_{\bbM^n}^\perp$. The required orthogonal cospan
$(V_1\subseteq W)\perp (V_2\subseteq W)$ from Assumption \ref{assu:PFA} 
can be obtained from the following geometric construction: Use that $U_i =  I^-(\{p_i^+\})\cap I^+(\{p_i^-\})$,
for some points $p^\pm_i\in \bbM^n$, and consider the spacelike line segment 
$L:=\big\{p_1^\circ+s(p^\circ_2 - p^\circ_1)\,:\, s\in [0,1] \big\}\subseteq\bbM^n$
between the two center points $p_i^\circ\in U_i$. This line segment
intersects the light cone $J^\mp(\{p_i^\pm\})\setminus I^\mp(\{p_i^\pm\})\subseteq \bbM^n$ 
precisely once and we denote by $q_i^\pm\in \bbM^n$
the intersection point. By construction, $p_i^\pm-q_i^\pm\in\bbR^n$ is a future/past pointing
null vector and the half-line $H^\pm_i :=\big\{q_i^\pm + s(p_i^\pm-q_i^\pm)\,:\, s\in[0,\infty)\big\}\subseteq \bbM^n$
passes through $p_i^\pm$ and eventually enters the complement $\bbM^n\setminus \mathrm{cl}(\widetilde{U})$
of the closure of the double cone $\widetilde{U}\subseteq \bbM^n$.
Picking any points $l^\pm_i\in H^\pm_i \cap \big(\bbM^n\setminus \mathrm{cl}(\widetilde{U})\big)$ in the complement, we define 
the double cones
\begin{flalign}
V_i \,:=\,  I^-(\{l_i^+\})\cap I^+(\{l_i^-\})\,\subseteq\,\bbM^n\quad,
\end{flalign}
which by construction are causally disjoint and satisfy $U_i\subseteq V_i$, for $i=1,2$.
By filteredness of $\mathbf{DCone}_{\bbM^n}^\perp$ there exists a 
double cone $W\subseteq \bbM^n$ such that $\widetilde{U}\subseteq W$ and $(V_1\subseteq W)\perp (V_2\subseteq W)$. 
Furthermore, the intersections $V_i\cap \big(\bbM^n\setminus J\big(\mathrm{cl}(\widetilde{U})\big)\big)\neq \varnothing$
of these $V_i$ with the causal complement of $\mathrm{cl}(\widetilde{U})\subseteq \bbM^n$ are by construction
non-empty and open, so there exist double cones $U_i^\prime 
\subseteq V_i\cap \big(\bbM^n\setminus J\big(\mathrm{cl}(\widetilde{U})\big)\big)$
such that $(U_1^\prime\subseteq W)\perp (\widetilde{U}\subseteq W)$
and $(U_2^\prime\subseteq W)\perp (\widetilde{U}\subseteq W)$. This 
provides all the required data in Assumption \ref{assu:PFA}.
\end{proof}

Given any $\AAA\in\AQFT\big(\mathbf{DCone}_{\bbM^n}^\perp\big)$ that satisfies
Haag duality in a locally faithful vacuum representation $\pi_0 : \AAA_\star\to \mathrm{B}(H_0)$,
the hypotheses of Theorem \ref{theo:PFAstructure} are met, so we obtain a locally constant
prefactorization algebra on double cones in $\bbM^n$ 
\begin{flalign}\label{eqn:SSSDcone}
\SSS_{(\AAA,\pi_0)}\,\in\,\Alg_{\P_{\mathbf{DCone}_{\bbM^n}^\perp}}^{}\!\Big(\Alg_{\mathsf{uAs}}\big(\CastCat\big)\Big)
\end{flalign}
which describes its superselection sector theory. We will now show that
this algebraic structure on the Minkowski spacetime $\bbM^n = \bbR^1\times\Sigma^{n-1}$ admits,
through the affine projection map $\pi : \bbM^n\to \Sigma^{n-1}$, an equivalent but simpler
description on the spatial factor $\Sigma^{n-1}=\bbR^{n-1}$, which
is governed by the following orthogonal category.
\begin{defi}\label{def:Convast}
We denote by $\mathbf{Conv}^\perp_{\Sigma^{n-1}}$ the orthogonal category whose objects are
all non-empty convex open subsets $S\subseteq \Sigma^{n-1}$, each endowed
with a distinguished point $q^\circ\in S$, morphisms are all subset inclusions $S\subseteq S^\prime$
(not necessarily preserving the distinguished points), and orthogonality relation is given by disjointness, i.e.\
$(S_1\subseteq T)\perp (S_2\subseteq T)$ if and only if $S_1\cap S_2=\emptyset$.
\end{defi}

\begin{lem}\label{lem:doublecone2conv}
The affine projection map $\pi : \bbM^n\to \Sigma^{n-1}$ defines an orthogonal functor
\begin{flalign}\label{eqn:doublecone2conv}
\pi\,:\, \mathbf{DCone}_{\bbM^n}^\perp~\longrightarrow~\mathbf{Conv}_{\Sigma^{n-1}}^{\perp}
\end{flalign}
which assigns to each double cone $U = I^-(\{p^+\})\cap I^+(\{p^-\})\subseteq \bbM^n$
its image $\pi(U)\subseteq \Sigma^{n-1}$,
endowed with the distinguished point $\pi(p^\circ)\in \pi(U)$ which is obtained
from the center point $p^\circ := p^- + \frac{1}{2}(p^+-p^-)\in U$ of the double cone.
\end{lem}
\begin{proof}
The underlying functor is well-defined because double cones are in particular 
non-empty convex open subsets, and these properties are 
preserved under the affine projection map $\pi : \bbM^n\to \Sigma^{n-1}$.
It therefore remains to show that this functor preserves the orthogonality relations.
Given any orthogonal cospan $(U_1\subseteq V)\perp (U_2\subseteq V)$ in 
$\mathbf{DCone}_{\bbM^n}^\perp$, one has by definition  
that $U_1$ and $U_2$ are causally disjoint subsets of  $\bbM^n$. This implies that in our
signature $(-+\cdots+)$ the squared Minkowski distance $d_{\bbM^n}^2(p_1,p_2)>0$ is positive, 
for all $p_1\in U_1$ and $p_2\in U_2$. Using further that
the squared Euclidean distance on $\Sigma^{n-1}$ satisfies
$d_{\bbM^n}^2(p,p^\prime)\leq d_{\Sigma^{n-1}}^2(\pi(p),\pi(p^\prime))$, for
all $p,p^\prime\in \bbM^n$, we conclude that $\pi(U_1)\cap \pi(U_2)=\varnothing$
and hence $(\pi(U_1)\subseteq \pi(V))\perp (\pi(U_2)\subseteq \pi(V))$ in 
$\mathbf{Conv}_{\Sigma^{n-1}}^\perp$.
\end{proof}

The orthogonal functor from Lemma \ref{lem:doublecone2conv} allows
us to compare prefactorization algebras on the Minkowski spacetime $\bbM^n=\bbR^1\times \Sigma^{n-1}$
with those on the spatial factor $\Sigma^{n-1}$. To facilitate
a comparison which also takes into account the crucial local constancy property
from Theorem \ref{theo:PFAstructure}, it will be convenient to embed 
our particular problem into Lurie's general theory of $\infty$-categories and $\infty$-operads \cite{LurieHA}.
As explained in more detail in \cite[Section 4.1]{BCNSsectors}, the orthogonal
functor from Lemma \ref{lem:doublecone2conv} defines a pullback $\infty$-functor
\begin{flalign}\label{eqn:pullback}
\pi^\ast\,:\, \AAlg_{\P_{\mathbf{Conv}_{\Sigma^{n-1}}^\perp}}^{\mathrm{l.c.}}\!\Big(\AAlg_{\bbE_1}\big(\CCastCat\big)\Big)~\longrightarrow~\AAlg_{\P_{\mathbf{DCone}_{\bbM^n}^\perp}}^{\mathrm{l.c.}}\!\Big(\AAlg_{\bbE_1}\big(\CCastCat\big)\Big)
\end{flalign}
between the corresponding $\infty$-categories of locally constant prefactorization algebras,
where $\CCastCat$ denotes the symmetric monoidal $\infty$-category
obtained from the combinatorial simplicial symmetric monoidal model structure on $\CastCat$ from \cite{CastCat}
and $\AAlg_{\bbE_1}\big(\CCastCat\big)$ denotes the symmetric monoidal $\infty$-category of $\bbE_1$-algebras in $\CCastCat$.
(The latter provides the $\infty$-categorical enhancement of the symmetric monoidal category
$\Alg_{\mathsf{uAs}}\big(\CastCat\big)$ of strict monoidal $C^\ast$-categories.)
Our locally constant prefactorization algebra of superselection sectors from Theorem \ref{theo:PFAstructure},
see also \eqref{eqn:SSSDcone}, defines an object 
\begin{flalign}\label{eqn:SSSDconeinfty}
\SSS_{(\AAA,\pi_0)}\,\in\,
\AAlg_{\P_{\mathbf{DCone}_{\bbM^n}^\perp}}^{\mathrm{l.c.}}\!\Big(\AAlg_{\bbE_1}\big(\CCastCat\big)\Big)
\end{flalign}
in the codomain of this $\infty$-functor. 
\begin{theo}\label{theo:DCone}
The pullback $\infty$-functor \eqref{eqn:pullback} is an equivalence of $\infty$-categories.
This provides an equivalent description of the locally constant prefactorization algebra
of superselection sectors \eqref{eqn:SSSDconeinfty} in terms of a locally constant prefactorization algebra 
\begin{flalign}\label{eqn:SSSConvinfty}
\SSS^{\Sigma^{n-1}}_{(\AAA,\pi_0)}\,\in\,
\AAlg_{\P_{\mathbf{Conv}_{\Sigma^{n-1}}^\perp}}^{\mathrm{l.c.}}\!\Big(\AAlg_{\bbE_1}\big(\CCastCat\big)\Big)
\end{flalign}
on the spatial factor $\Sigma^{n-1}$ of the Minkowski spacetime $\bbM^n=\bbR^1\times\Sigma^{n-1}$.
\end{theo}

A proof of Theorem \ref{theo:DCone} will be presented in Appendix \ref{app:technical}, 
because it involves techniques and results from the theory of 
$\infty$-categories and $\infty$-operads, which in our opinion would distract
the reader from the core idea. Let us however anticipate that 
the above mentioned techniques formalize the intuitive idea 
that local constancy allows one to shrink and enlarge the open subsets 
of the corresponding prefactorization operads, and also move them around
in a way that respects the orthogonality relations. Since all double cones $U\subseteq \bbM^n$ in
$\mathbf{DCone}_{\bbM^n}^\perp$ and all convex opens $S\subseteq \Sigma^{n-1}$ 
in $\mathbf{Conv}_{\Sigma^{n-1}}^\perp$ come endowed with a distinguished point,
see Remark \ref{rem:doublecone} and Definition \ref{def:Convast}, one can 
shrink such subsets to their distinguished points, leading to 
configuration spaces of (causally) disjoint points. Summing up, the theory of 
$\infty$-categories and $\infty$-operads allows us to reduce the proof 
of Theorem \ref{theo:DCone} to the comparison from Proposition \ref{prop:Confhomotopy} below
between the following configuration spaces of (causally) disjoint points.
\begin{defi}\label{def:conf}
\begin{itemize}
\item[(a)] The \textit{configuration space of $m\in\bbZ_{\geq 0}$ disjoint points in 
$S\in \mathbf{Conv}_{\Sigma^{n-1}}^{\perp}$} is defined as the subset
\begin{flalign}
\mathrm{Conf}_m(S)\,:=\,\Big\{\und{q} = (q_1,\dots, q_m) \in S^{\times m} 
\,:\, q_i \neq q_j~\forall \,i\neq j \Big\}\,\subseteq\,S^{\times m}\quad,
\end{flalign}
endowed with the subspace topology.

\item[(b)] The \textit{configuration space of $m\in\bbZ_{\geq 0}$ causally disjoint points in 
$U\in \mathbf{DCone}_{\bbM^{n}}^{\perp}$} is defined as the subset
\begin{flalign}
\mathrm{cConf}_m(U)\,:=\,\Big\{\und{p} = (p_1,\dots, p_m) \in U^{\times m} 
\,:\,p_j\not\in J(\{p_i\})~\forall \,i\neq j \Big\}\,\subseteq\,U^{\times m}\quad,
\end{flalign}
endowed with the subspace topology.
\end{itemize}
\end{defi}

For every $m\in\bbZ_{\geq 0}$ and every double cone $U\in \mathbf{DCone}_{\bbM^n}^\perp$, 
the projection map $\pi : \bbM^n\to \Sigma^{n-1}$ defines a continuous map 
between configuration spaces
\begin{flalign}\label{eqn:Confpi}
\pi\,:\, \mathrm{cConf}_{m}(U)~\longrightarrow~\mathrm{Conf}_m(\pi(U))~,~~
\und{p}~\longmapsto~\pi(\und{p}) \,=\,\big(\pi(p_1),\dots,\pi(p_m)\big)\quad.
\end{flalign}
Our goal is to show that this map is a homotopy equivalence, which we 
achieve by constructing an explicit quasi-inverse.
\begin{constr}\label{constr:quasiinverse}
We start with observing that every double cone $U = I^-(\{p^+\})\cap I^+(\{p^-\}) 
\in \mathbf{DCone}_{\bbM^n}^\perp$ can be endowed with a canonical Cauchy surface
defined by
\begin{flalign}
\Sigma_U\,:=\, \Big\{p \in  U\,:\, \big\langle p-p^\circ , p^+-p^-\big\rangle_{\bbM^n}  = 0 \Big\}\,\subset\, U\quad,
\end{flalign}
i.e.\ it is the set of all points $p\in U$ whose difference vector $p-p^\circ\in\bbR^n$ from the center
$p^\circ \in U$ is normal, with respect to the Minkowski 
inner product $\langle\cdot,\cdot\rangle_{\bbM^n}$, to the timelike vector $p^+-p^-\in \bbR^n$. 
Restricting the projection map 
\eqref{eqn:Confpi} to configurations in this Cauchy surface defines an isomorphism of topological spaces
\begin{flalign}
\xymatrix{
\pi\vert\,:\,\mathrm{cConf}_{m}(U)\cap \Sigma_U^{\times m} \ar[r]^-{\cong}
~&~\mathrm{Conf}_m(\pi(U))
}\quad.
\end{flalign}
We use its inverse to define the continuous map
\begin{flalign}\label{eqn:Confiota}
\xymatrix{
\iota\,:\, \mathrm{Conf}_m(\pi(U)) \ar[r]^-{\pi\vert^{-1}}_-{\cong}~&~
\mathrm{cConf}_{m}(U)\cap \Sigma_U^{\times m}\ar@{^{(}->}[r]~&~\mathrm{cConf}_{m}(U)
}\quad,
\end{flalign}
which by construction satisfies $\pi\circ \iota = \id_{\mathrm{Conf}_m(\pi(U))}$.
The other composition $\iota\circ \pi$ is clearly not equal to the identity
because it projects any configuration of causally disjoint points
$\und{p}\in \mathrm{cConf}_{m}(U)$ to a configuration $\iota\pi(\und{p})\in \mathrm{cConf}_{m}(U)\cap \Sigma_U^{\times m}$
of (causally) disjoint points in the Cauchy surface $\Sigma_U\subset U$.
Instead, as we shall prove below in Proposition \ref{prop:Confhomotopy},
one has that $\iota\circ \pi\sim\id_{\mathrm{cConf}_{m}(U)}$
are homotopic through the homotopy
\begin{subequations}\label{eqn:Confhomotopy}
\begin{flalign}
h\,:\,[0,1]\times\mathrm{cConf}_{m}(U) ~\longrightarrow~\mathrm{cConf}_{m}(U)
\end{flalign}
defined by
\begin{flalign}
h(s,\und{p})\,:=\,
\Big(\iota\pi(p_1) + s\big(p_1-\iota\pi(p_1)\big),\dots,\iota\pi(p_m) + s\big(p_m-\iota\pi(p_m)\big)\Big)\quad.
\end{flalign}
\end{subequations}
The idea behind this homotopy is to recover the original 
points $p_i\in U$ from their projections $\iota\pi(p_i)\in \Sigma_U$ to the Cauchy surface
by translating along the timelike vectors $p_i-\iota\pi(p_i)\in\bbR^n$.
\end{constr}

\begin{propo}\label{prop:Confhomotopy}
The continuous map $\iota$ in \eqref{eqn:Confiota} and the homotopy $h$ in 
\eqref{eqn:Confhomotopy} define a quasi-inverse of the continuous map $\pi$ in \eqref{eqn:Confpi}. 
It then follows in particular that \eqref{eqn:Confpi} is a homotopy equivalence of topological spaces.
\end{propo}
\begin{proof}
It remains to show that the homotopy \eqref{eqn:Confhomotopy} is well-defined
because by definition it satisfies the desired identities 
$h(0,\,\cdot\,) = \iota\circ \pi$ and $h(1,\,\cdot\,) = \id_{\mathrm{cConf}_{m}(U)}$
leading to $\iota\circ \pi\sim\id_{\mathrm{cConf}_{m}(U)}$. (The property
$\pi\circ \iota = \id_{\mathrm{Conf}_m(\pi(U))}$ was already observed above.)
To prove well-definedness, we have to show that $h(s,\und{p})\in \mathrm{cConf}_{m}(U)$
is a configuration of causally disjoint points, for all $\und{p}\in \mathrm{cConf}_{m}(U)$
and all $s\in[0,1]$. Using the squared Minkowski distance to detect causal disjointness, 
this amounts to showing that
\begin{flalign}
\nn & d_{\bbM^n}^2\Big(\iota\pi(p_i) + s\big(p_i-\iota\pi(p_i)\big),\iota\pi(p_j) + s\big(p_j-\iota\pi(p_j)\big)\Big)\\[4pt]
&\qquad = \Big\vert\!\Big\vert (1-s) \big(\iota\pi(p_i) - \iota\pi(p_j)\big) + s \big(p_i-p_j\big)\Big\vert\!\Big\vert_{\bbM^n}^2 \,>\,0
\label{tmp:estimate}
\end{flalign}
is positive, for all $i\neq j$ and all $s\in[0,1]$. Introducing for notational convenience
the two spacelike vectors $v:=\iota\pi(p_i) - \iota\pi(p_j)\in \bbR^n$ and $w:=p_i-p_j\in \bbR^n$,
one can write this expression as
\begin{flalign}
\nn \big\vert\!\big\vert (1-s)v + s w\big\vert\!\big\vert_{\bbM^n}^2 \,&=\,
(1-s)^2 \,\vert\!\vert v \vert\!\vert_{\bbM^n}^{2}  + 2s(1-s)\,\langle v,w\rangle_{\bbM^n}^{}
+ s^2\,\vert\! \vert w \vert\!\vert_{\bbM^n}^{2}\\
\,&>\, 2s(1-s)\,\langle v,w\rangle_{\bbM^n}^{} \quad,
\end{flalign}
where in the last step we used that 
$\vert\!\vert v \vert\!\vert_{\bbM^n}^{2}>0$ and $\vert\!\vert w \vert\!\vert_{\bbM^n}^{2}>0$.
Using now that the spatial projections $\pi(v) = \pi(w)$ of the two spacelike vectors
coincide, which is a direct consequence
of $\pi\circ\iota = \id_{\mathrm{Conf}_m(\pi(U))}$, it follows that the inner product 
$\langle v,w\rangle_{\bbM^n}^{} > 0$ is positive too,
leading to a proof of the inequality \eqref{tmp:estimate}.
\end{proof}

To identify the equivalent prefactorization algebra $\SSS^{\Sigma^{n-1}}_{(\AAA,\pi_0)}$
on the spatial factor $\Sigma^{n-1}$ from Theorem \ref{theo:DCone}
with an $\bbE_{n-1}$-algebra, we have to understand how our prefactorization
operad $\P_{\mathbf{Conv}_{\Sigma^{n-1}}^\perp}$ relates to an 
$\bbE_{n-1}$-operad. A concrete and useful model for the latter
is given by the standard \textit{little $(n-1)$-disks operad}, see e.g.\ \cite{Horel}.
In our context, the latter can be defined by fixing any point $q^\circ\in \Sigma^{n-1}$
in the spatial factor $\Sigma^{n-1}\cong \bbR^{n-1}$ of the Minkowski spacetime
in order to regard $D:=\Sigma^{n-1}\in \mathbf{Conv}_{\Sigma^{n-1}}^\perp$
as an object in the orthogonal category from Definition \ref{def:Convast}. 
A model for the $\bbE_{n-1}$-operad is then given by the topological operad
consisting of the single object $D$ and $m$-ary operation spaces
given by the framed embedding spaces $\mathrm{Emb}_f\big(D^{\sqcup m},D\big)$.
The link between our discrete operad $\P_{\mathbf{Conv}_{\Sigma^{n-1}}^\perp}$ and this
topological model for the $\bbE_{n-1}$-operad is given by the following intermediate topological
operad which is based on \textit{all} non-empty convex opens in $\Sigma^{n-1}$
as in Definition \ref{def:Convast}.
\begin{defi}\label{def:PConvtop}
We denote by $\P_{\CConv_{\Sigma^{n-1}}^\perp}$ the topological suboperad
of the topological operad from \cite[Definition 8.1]{Horel} whose objects
are all objects $S\in \mathbf{Conv}_{\Sigma^{n-1}}^\perp$ of the 
orthogonal category from Definition \ref{def:Convast} and whose operations are
given by the framed embedding spaces $\P_{\CConv_{\Sigma^{n-1}}^\perp}\big(\substack{T\\ \und{S}}\big)
:=\mathrm{Emb}_f\big(S_1\sqcup\cdots\sqcup S_m,T\big)$.
\end{defi}

By definition, there exist morphisms of topological operads
\begin{flalign}\label{eqn:zigzag}
\xymatrix{
\P_{\mathbf{Conv}_{\Sigma^{n-1}}^\perp} \ar[r]~&~
\P_{\CConv_{\Sigma^{n-1}}^\perp} ~&~\ar[l] \bbE_{n-1}
}\quad,
\end{flalign}
where $\P_{\mathbf{Conv}_{\Sigma^{n-1}}^\perp}$ is considered
as a topological operad with discrete operation spaces.
The right-pointing morphism acts as the identity on objects 
and it assigns to every mutually disjoint inclusion $\und{S}=(S_1,\dots,S_m) \to T$ 
of subsets of $\Sigma^{n-1}\cong\bbR^{n-1}$, i.e.\ $S_{i}\subseteq T$ with $S_i\cap S_j =\varnothing$ for all $i\neq j$, 
its corresponding framed embedding $S_1\sqcup\cdots\sqcup S_m\to T$. The left-pointing morphism
fully faithfully embeds the $\bbE_{n-1}$-operad by regarding the spatial factor
$D=\Sigma^{n-1}\in  \mathbf{Conv}_{\Sigma^{n-1}}^\perp$ as a non-empty convex open subset
with the chosen distinguished point $q^\circ \in D$. The proof of the following
result is presented in Appendix \ref{app:technical}.
\begin{theo}\label{theo:Enidentification}
The morphisms of topological operads \eqref{eqn:zigzag} induce pullback $\infty$-functors
\begin{flalign}\label{eqn:AlgsoverConv}
\resizebox{0.91\hsize}{!}{
\xymatrix@C=1.5em{
\AAlg_{\P_{\mathbf{Conv}_{\Sigma^{n-1}}^\perp}}^{\mathrm{l.c.}}\!\Big(\AAlg_{\bbE_1}\big(\CCastCat\big)\Big)~&~
\ar[l]_-{\sim}\AAlg_{\P_{\CConv_{\Sigma^{n-1}}^\perp}}\!\Big(\AAlg_{\bbE_1}\big(\CCastCat\big)\Big)\ar[r]^-{\sim}~&~
\AAlg_{\bbE_{n-1}}\!\Big(\AAlg_{\bbE_1}\big(\CCastCat\big)\Big)
}}
\end{flalign}
which are equivalences of $\infty$-categories. Together with Theorem \ref{theo:DCone} and
Dunn-Lurie additivity $\AAlg_{\bbE_{n-1}}\!\big(\AAlg_{\bbE_1}\big(\CCastCat\big)\big) \simeq
\AAlg_{\bbE_{n}}\big(\CCastCat\big)$,
this provides an equivalent description of the locally constant prefactorization algebra
of superselection sectors \eqref{eqn:SSSDconeinfty} in terms of an $\mathbb{E}_n$-algebra
\begin{flalign}
\SSS^{D}_{(\AAA,\pi_0)}\,\in\, \AAlg_{\bbE_{n}}\big(\CCastCat\big)\quad,
\end{flalign}
i.e.\ an $\bbE_n$-monoidal $C^\ast$-category.
\end{theo}

\begin{rem}\label{rem:identification}
From this theorem one recovers the standard result \cite{DHR,Halvorson} 
that the $C^\ast$-category of superselection sectors
of an AQFT which is defined on double cones in the 
$(n\geq 2)$-dimensional Minkowski spacetime $\bbM^n$ is
braided monoidal in dimension $n=2$
and symmetric monoidal in dimension $n\geq 3$. The latter degenerate behavior follows
from the fact that $\CCastCat$ is only a $2$-category, hence
$\AAlg_{\bbE_{n}}\big(\CCastCat\big) \simeq \AAlg_{\bbE_{\infty}}\big(\CCastCat\big)$
for all $n\geq 3$. The novel aspect of our approach is that it explains the 
conceptual and geometric origins of this result: The $\bbE_1$ factor
in Dunn-Lurie additivity $\bbE_{n}\simeq \bbE_1\otimes \bbE_{n-1}$ has an analytic origin
in Haag duality and the associated object-wise monoidal structures
from Proposition \ref{prop:monoidal}. In contrast, the $\bbE_{n-1}$ factor
has a geometric origin arising from the interplay between the general
locally constant prefactorization algebra structure from Theorem 
\ref{theo:PFAstructure} and example-specific homotopy theoretical results
about double cones in $\bbM^n = \bbR^1\times\Sigma^{n-1}$ and their projections to the spatial factor $\Sigma^{n-1}$,
leading in particular to the key Proposition \ref{prop:Confhomotopy}. 
Hence, this second type of contribution to the algebraic structures on 
the $C^\ast$-categories of superselection sectors depends on the geometry 
and topology of the orthogonal category $\CC^\perp$ on which the AQFTs are defined. 
Concrete examples of how this can alter or refine these algebraic structures
have been described in \cite{BCNSsectors} in the context of lattice quantum systems. 
In the present context, one can show that enlarging the orthogonal category 
of double cones $\mathbf{DCone}_{\bbM^n}^\perp$ to the orthogonal category of diamonds 
$\mathbf{RC}_{\bbM^n}^{\diamond\, \perp}$ (i.e.\ the poset of relatively 
compact causally convex open subsets of $\bbM^n$ that are 
diffeomorphic to $\bbR^n$ endowed with the orthogonality relation 
given by causal disjointness) leads to the same conclusions.
\end{rem}

We conclude this section with a brief observation about an evident 
$G$-equivariant generalization of the above results. 
Consider a discrete group $G$ acting through isometries on 
the spatial factor $\Sigma^{n-1}$ of the Minkowski 
spacetime $\bbM^n=\bbR^1\times\Sigma^{n-1}$. This induces $G$-actions
\begin{subequations}\label{eqn:operadGactions}
\begin{flalign}
\P_{\mathbf{DCone}_{\bbM^n}^\perp}^{} \,&:\, \mathsf{B}G~\longrightarrow~\mathbf{Op}\quad,\\
\P_{\mathbf{Conv}_{\Sigma^{n-1}}^\perp}^{} \,&:\, \mathsf{B}G~\longrightarrow~\mathbf{Op}\quad,\\
\P_{\CConv_{\Sigma^{n-1}}^\perp}^{}\,&:\,\mathsf{B}G~\longrightarrow~\mathbf{Op}\quad,\\
\bbE_{n-1}\,&:\,\mathsf{B}G~\longrightarrow~\mathbf{Op}
\end{flalign}
\end{subequations}
on the (topological) operads introduced above.
Since all comparison
operad morphisms in \eqref{eqn:doublecone2conv} and \eqref{eqn:zigzag}
preserve these $G$-actions, they induce pullback
$\infty$-functors which allow us to compare their corresponding
$\infty$-categories of $G$-equivariant (locally constant) algebras,
which are defined in analogy to Definition \ref{def:eqvPFA} by replacing
the bicategorical limit with an $\infty$-categorical one.
Then, Theorems \ref{theo:DCone} and \ref{theo:Enidentification}
imply directly the following result.
\begin{cor}\label{cor:eqvexample}
Suppose that a discrete group $G$ acts through isometries on 
the spatial factor $\Sigma^{n-1}$ of the $(n\geq 2)$-dimensional Minkowski 
spacetime $\bbM^n=\bbR^1\times\Sigma^{n-1}$ and consider the induced
$G$-actions \eqref{eqn:operadGactions} on operads. Then, the equivalence
from Theorem \ref{theo:DCone} induces an equivalence
\begin{flalign}
\AAlg_{\P_{\mathbf{DCone}_{\bbM^n}^\perp}}^{\mathrm{l.c.},\,G}\!\Big(\AAlg_{\bbE_1}\big(\CCastCat\big)\Big)
~\simeq~
\AAlg_{\P_{\mathbf{Conv}_{\Sigma^{n-1}}^\perp}}^{\mathrm{l.c.},\,G}\!\Big(\AAlg_{\bbE_1}\big(\CCastCat\big)\Big)
\end{flalign}
between the corresponding $\infty$-categories of $G$-equivariant locally constant prefactorization algebras.
Furthermore, the equivalence from Theorem \ref{theo:Enidentification} induces an equivalence
\begin{flalign}
\AAlg_{\P_{\mathbf{Conv}_{\Sigma^{n-1}}^\perp}}^{\mathrm{l.c.},\,G}\!\Big(\AAlg_{\bbE_1}\big(\CCastCat\big)\Big)
~\simeq~
\AAlg_{\bbE_{n-1}}^{G}\!\Big(\AAlg_{\bbE_1}\big(\CCastCat\big)\Big)
~\simeq~
\AAlg_{\bbE_{n}^G}^{}\big(\CCastCat\big)
\quad,
\end{flalign}
where the last equivalence follows from \cite{Arakawa}.
Here $\bbE_{n}^G:= G\ltimes\bbE_n$ denotes
the semi-direct product operad from \cite{Markl,Wahl}, which in our case 
is obtained from the induced $G$-action on $\bbR^1\times \Sigma^{n-1}$.
\end{cor}
\begin{rem}
The $G$-equivariant locally constant prefactorization algebras
of (covariant) superselection sectors from Theorems \ref{theo:SSSeqvPFA} and \ref{theo:covSSSeqvPFA}
define objects
\begin{flalign}
\SSS_{(\AAA,\pi_0)}^{(\mathrm{cov})}\,\in\,\AAlg_{\P_{\mathbf{DCone}_{\bbM^n}^\perp}}^{\mathrm{l.c.},\,G}\!
\Big(\AAlg_{\bbE_1}\big(\CCastCat\big)\Big)\quad.
\end{flalign}
(Note that these objects are defined in the more general case of $G$ 
acting through arbitrary isometries of the Minkowski spacetime $\bbM^n$.)
According to Corollary \ref{cor:eqvexample}, the latter can be identified with
$G$-equivariant locally constant prefactorization algebras 
\begin{flalign}
\SSS^{\Sigma^{n-1},\,(\mathrm{cov})}_{(\AAA,\pi_0)}\,\in\,
\AAlg_{\P_{\mathbf{Conv}_{\Sigma^{n-1}}^\perp}}^{\mathrm{l.c.},\,G}\!\Big(\AAlg_{\bbE_1}\big(\CCastCat\big)\Big)
\end{flalign}
on the spatial factor $\Sigma^{n-1}$ of the Minkowski spacetime $\bbM^n$,
which admit a further equivalent description
in terms of $G$-equivariant $\bbE_n$-monoidal $C^\ast$-categories
\begin{flalign}
\SSS^{D,\,(\mathrm{cov})}_{(\AAA,\pi_0)}\,\in\,\AAlg_{\bbE_{n-1}}^{G}\!\Big(\AAlg_{\bbE_1}\big(\CCastCat\big)\Big)\,
\simeq\, \AAlg_{\bbE_{n}^G}^{}\big(\CCastCat\big)\quad.
\end{flalign}
(For these identifications it is instead crucial that $G$ acts through isometries of the spatial factor $\Sigma^{n-1}$.)
The physical interpretation of this result is that a group $G$ of 
space(time) symmetries of an AQFT $\AAA$ defined on double cones in $\bbM^n$
induces a $G$-action on its associated $\bbE_n$-monoidal $C^\ast$-category of (covariant) superselection sectors. 
\end{rem}

\begin{rem} 
Note that in the case $G=\mathrm{SO}(n-1)$, the topological operad $\bbE_{n-1}^G$ 
we consider above is \textit{not} equivalent to the framed $\bbE_{n-1}$-operad from \cite{Markl,Wahl}.
This stems from the fact that we are regarding $\mathrm{SO}(n-1)$ as a discrete group,
ignoring its Lie group structure. 
\end{rem}


\section*{Acknowledgments}
The work of M.B.\ is supported in part by the MUR Excellence 
Department Project awarded to Dipartimento di Matematica, 
Universit{\`a} di Genova (CUP D33C23001110001) and it is fostered by 
the National Group of Mathematical Physics (GNFM-INdAM (IT)). 
V.C.\ is partially supported by the project PID2024-157173NB-I00 funded
by MCIN/AEI/10.13039/501100011033 and by FEDER, UE.


%
%


\appendix

\section{\label{app:technical}Technical details for Theorems \ref{theo:DCone} and \ref{theo:Enidentification}}
In this appendix, we provide the missing $\infty$-categorical arguments 
which demonstrate the equivalences of $\infty$-categories of operadic algebras 
stated in Theorems \ref{theo:DCone} and \ref{theo:Enidentification}.
The main technical tool we employ is that of \textit{weak approximations} of Lurie's $\infty$-operads, 
which provide a suitable criterion in our setting to detect $\infty$-localizations 
of $\infty$-operads. The original source is \cite[Proposition 2.3.3.11 and Theorem 2.3.3.23]{LurieHA}, 
but the variation that we will use can be found in \cite[Section 4.2]{Harpaz}.
A summary of the latter appears also in \cite[Appendix C]{CarmonaSvraka}, 
see in particular Definition C.4 and Corollary C.16 of this work.
\begin{rem}\label{rem:generalsmcatinsteadofAAlg} 
Since the arguments below do not depend on the target 
category in which our prefactorization algebras take values,
we can and will replace $\AAlg_{\bbE_1}\big(\CCastCat\big)$ by a general symmetric monoidal $\infty$-category $\V$ 
in the rest of this appendix. In this way we can also alleviate our notations.
\end{rem}

First of all, let us observe that, even for a general symmetric monoidal $\infty$-category
$\V$ in place of $\AAlg_{\bbE_1}\big(\CCastCat\big)$,
the right-pointing equivalence in \eqref{eqn:AlgsoverConv} is a direct consequence of 
the fact that the canonical inclusion of topological operads
$\mathbb{E}_{n-1}\rightarrow \P_{\CConv_{\Sigma^{n-1}}^\perp}$
yields an equivalence of $\infty$-operads, see \cite[Proposition 4.1.23]{Harpaz}. 
Indeed, it is essentially surjective up to homotopy since $\mathrm{Emb}_f(S,S^\prime)$ 
is weakly contractible, for any $S,S^\prime\in \mathbf{Conv}_{\Sigma^{n-1}}$, see e.g.\ 
\cite[Proposition 7.6]{Horel}, and homotopically fully-faithful by definition.
\sk

The remaining $\infty$-functors between $\infty$-categories of operadic
algebras appearing in Theorems \ref{theo:DCone} and \ref{theo:Enidentification} 
fit into the commutative triangle
\begin{equation}
\begin{tikzcd}
\AAlg_{\P_{\mathbf{DCone}_{\mathbb{M}^{n}}^\perp}}^{\mathrm{l.c.}}\!\big(\V\big)  && 
\AAlg_{\P_{\mathbf{Conv}_{\Sigma^{n-1}}^\perp}}^{\mathrm{l.c.}}\!\big(\V\big)  \ar[ll,"\pi^\ast"']\\
& \AAlg_{\P_{\CConv_{\Sigma^{n-1}}^\perp}}\!\big(\V\big)  \ar[ru] \ar[lu]
\end{tikzcd}\quad,
\end{equation}
hence it suffices to verify that both diagonal $\infty$-functors 
are equivalences for achieving the sought conclusions. This problem 
is equivalent to proving that the two maps of topological operads
\begin{flalign}\label{eqn:Opdmapstobeweakapprox}
\P_{\mathbf{DCone}_{\mathbb{M}^{n}}^\perp} ~\stackrel{\pi}{\longrightarrow}~
\P_{\mathbf{Conv}_{\Sigma^{n-1}}^\perp}~\longrightarrow~\P_{\CConv_{\Sigma^{n-1}}^\perp}\quad,\qquad
\P_{\mathbf{Conv}_{\Sigma^{n-1}}^\perp} ~\longrightarrow~\P_{\CConv_{\Sigma^{n-1}}^\perp}
\end{flalign}
exhibit $\P_{\CConv_{\Sigma^{n-1}}^\perp}$ as the $\infty$-localization 
at all unary operations of the two operads 
$\P_{\mathbf{DCone}_{\mathbb{M}^{n}}^\perp}$ and $\P_{\mathbf{Conv}_{\Sigma^{n-1}}^\perp}$.
In light of \cite[Theorem 2.3.3.23]{LurieHA} or \cite[Corollary C.16]{CarmonaSvraka}, 
this is a consequence of the following technical result.
\begin{propo}\label{prop:weakapproximations}
The two maps of topological operads in \eqref{eqn:Opdmapstobeweakapprox} 
induce weak approximations\footnote{Note that Lurie's and Harpaz's definitions of weak 
approximations are not technically equivalent. This subtlety does not affect 
our arguments and conclusions because the underlying $\infty$-category of 
$\P_{\CConv_{\Sigma^{n-1}}^{\perp}}$ is an $\infty$-groupoid.} 
between their associated $\infty$-operads.
\end{propo}
\begin{proof}
We start with observing that the argument 
from \cite[Lemma 5.4.5.11]{LurieHA} applies verbatim to the second operad map
$\P_{\mathbf{Conv}_{\Sigma^{n-1}}^\perp}\to \P_{\CConv_{\Sigma^{n-1}}^\perp}$,
see also \cite[Proposition 5.2.4]{Harpaz}. 
The proof for the first operad map $\P_{\mathbf{DCone}_{\mathbb{M}^{n}}^\perp} \to 
\P_{\mathbf{Conv}_{\Sigma^{n-1}}^\perp}\to
\P_{\CConv_{\Sigma^{n-1}}^\perp}$ will follow from a slight
variation of this argument involving the two types of configuration spaces from Definition \ref{def:conf}.
\sk

The first condition for a weak approximation,
see \cite[Definition 4.2.14]{Harpaz} or \cite[Definition C.4]{CarmonaSvraka}, 
is a consequence of the fact that the underlying $\infty$-category of 
$\P_{\CConv_{\Sigma^{n-1}}^\perp}$ is a contractible $\infty$-groupoid 
(recall that $\mathrm{Emb}_f(S,S^\prime)$ is weakly contractible, 
for any $S,S^\prime\in \mathbf{Conv}_{\Sigma^{n-1}}$), and that 
the underlying category of $\P_{\mathbf{DCone}_{\mathbb{M}^{n}}^\perp}$, 
which is given by $\mathbf{DCone}_{\mathbb{M}^{n}}$, 
is filtered by Lemma \ref{lem:doublecone}. 
\sk

Following a standard argument, see e.g.\ 
the proof of \cite[Lemma 5.4.5.11]{LurieHA}, 
\cite[Proposition 5.2.4]{Harpaz} or \cite[Lemma 3.6]{CarmonaSvraka}, 
the second condition for a weak approximation can be reduced to 
showing that the canonical map of topological spaces
\begin{flalign}
\underset{\underline{U}\in \mathbf{A}_V^m}{\hocolim}\bigg(\prod_{1\leq i\leq m} \mathrm{Emb}_f\big(S_i,\pi(U_i)\big)\bigg)
~\longrightarrow~ 
\mathrm{Emb}_f\Big(\bigsqcup_{1\leq i \leq m} S_i, \pi(V)\Big)
\end{flalign}
is a weak homotopy equivalence, for any given $V\in \mathbf{DCone}_{\mathbb{M}^n}$ 
and $m$-tuple $\underline{S}\in \mathbf{Conv}_{\Sigma^{n-1}}^{\times m}$.
The homotopy colimit in this expression is taken over the 
partially ordered set $\mathbf{A}_V^m \subseteq  \mathbf{DCone}_{\mathbb{M}^{n}}^{\times m}$ 
of all $m$-tuples of double cones $\und{U}=(U_1,\dots,U_m)$ 
satisfying the property that there exists a (unique) operation $\und{U}\to V$ 
in $\P_{\mathbf{DCone}_{\mathbb{M}^n}^{\perp}}$. 
By evaluating at the distinguished points of the $S_i$'s and 
using the standard homotopy theoretical result 
from \cite[Proposition 7.6]{Horel}, one can replace the spaces of framed embeddings 
above by (ordinary) configuration spaces, and alternatively prove that 
\begin{flalign}
\underset{\underline{U}\in \mathbf{A}_V^m}{\hocolim}\bigg(\prod_{1\leq i\leq m} \pi(U_i)\bigg)~
\longrightarrow~ \mathrm{Conf}_{m}\big(\pi(V)\big)
\end{flalign}
is a weak homotopy equivalence. We demonstrate that this is the case by considering the square
\begin{equation}\label{eqn:CommSquareforWeakApprox}
\begin{tikzcd}
\underset{\underline{U}\in \mathbf{A}_V^m}{\hocolim}\Big(\prod_{1\leq i\leq m} U_i\Big) 
\ar[rr]\ar[d,"\pi"'] && \mathrm{cConf}_{m}\big(V\big) \ar[d,"\pi"]\\
\underset{\underline{U}\in \mathbf{A}_V^m}{\hocolim}\Big(\prod_{1\leq i\leq m} \pi(U_i)\Big) 
\ar[rr] && \mathrm{Conf}_{m}\big(\pi(V)\big) 
 \end{tikzcd}\quad,
\end{equation}
whose commutativity is due to the fact that $\pi: \mathbb{M}^n\to \Sigma^{n-1}$ 
induces a natural transformation filling the  right triangle in 
\begin{equation}
\begin{tikzcd}
\underline{U} \ar[rr, mapsto] \ar[dd, mapsto] &[-9mm]& \prod_i U_i  \\[-7mm]
&\mathbf{A}^m_V \ar[r] \ar[d] &  \mathbf{Open}\big(\mathrm{cConf}_m\big(V\big)\big) \ar[d] \ar[ld,"\pi"'] \ar[d, phantom, "\Longleftarrow" description, shift right=12, near end]\\
\prod_{i}\pi(U_i)
& 
\mathbf{Open}\big(\mathrm{Conf}_m\big(\pi(V)\big)\big) \ar[r] & \EuScript{S}\mathsf{pc} 
\end{tikzcd}\quad,
\end{equation}
where $\EuScript{S}\mathsf{pc}$ denotes the $\infty$-category of spaces.
\sk

Concerning the diagram \eqref{eqn:CommSquareforWeakApprox}, we observe: 
\begin{itemize}
\item[(1)] The left vertical map is a weak homotopy equivalence because
convexity of the involved subsets implies that
$U_i\cong\mathbb{R}^{n}$ and $\pi(U_i)\cong \mathbb{R}^{n-1}$, for all $i=1,\dots,m$.

\item[(2)] The right vertical map is a homotopy equivalence by Proposition 
\ref{prop:Confhomotopy}.

\item[(3)] The upper horizontal map is a weak 
homotopy equivalence as a consequence of \cite[Theorem A.3.1]{LurieHA} since double cones 
provide a basis for the topology of $\mathbb{M}^n$. For 
this reason, for any $\underline{p}=(p_1,\dots,p_m)\in \mathrm{cConf}_m(V)$, the subposet
\begin{flalign}
\Big\{ \underline{U}\in \mathbf{A}^m_V\,:\, p_i\in U_i\,\text{ for all } i=1,\dots,m\Big\}\,\subseteq\, \mathbf{A}^m_V
\end{flalign}
is cofiltered and therefore weakly contractible. 
\end{itemize}
By the 2-out-of-3 property of weak homotopy equivalences, this implies that
the bottom horizontal map \eqref{eqn:CommSquareforWeakApprox} is a weak homotopy equivalence too,
which concludes the proof.
\end{proof}


\end{document}